\documentclass[10pt,twocolumn]{IEEEtran}
\IEEEoverridecommandlockouts
\usepackage{verbatim}
\usepackage{epsfig}
\usepackage{amsmath}
\usepackage{amsthm}
\usepackage{amssymb}
\usepackage{psfrag}
\usepackage[T1]{fontenc}
\usepackage[ruled]{algorithm2e}

\usepackage[usenames,dvipsnames]{pstricks}
\usepackage{pst-grad} 
\usepackage{pst-plot} 
\usepackage[ruled]{algorithm2e}
\usepackage{etoolbox}
\usepackage{subfig}
\makeatletter
\patchcmd{\@makecaption}
  {\scshape}
  {}
  {}
  {}
\makeatother
\usepackage[font=small,skip=0pt]{caption}
\usepackage[noadjust]{cite}

\newtheorem{theorem}{Theorem}
\newtheorem{lemma}{Lemma}

\newtheorem{definition}{Definition}
\newtheorem{assump}{Assumption}

\newtheorem{corollary}{Corollary}[theorem]
\newtheorem{remark}{Remark}

\newcommand{\E}{\mathcal{E}}
\newcommand{\V}{\mathcal{V}}
\newcommand{\G}{\mathcal{G}}

\newcommand{\x}{\mathbf{x}}
\newcommand{\y}{\mathbf{y}}
\newcommand{\z}{\mathbf{z}}
\newcommand{\p}{\mathbf{P}}
\newcommand{\I}{\mathbf{I}}
\newcommand{\A}{\mathbf{A}}
\newcommand{\B}{\mathbf{B}}

\newcommand{\M}{\mathbf{M}}

\newcommand{\nrps}{\textbf{\texttt{NR-PushSum}}}

\definecolor{color1}{rgb}{0,0,0}

\IEEEpeerreviewmaketitle

\title{Distributed Average Consensus Over Noisy Communication Links in Directed Graphs}
\author{Vivek Khatana$^1$ and Murti V. Salapaka$^{1}$ \\
\thanks{This work is supported by the Advanced Research Projects Agency-Energy OPEN through the project titled "Rapidly Viable Sustained Grid" via grant no. DE-AR0001016.}
\thanks{$^{1}$ Vivek Khatana \{{\tt\small khata010@umn.edu}\} and Murti V. Salapaka \{{\tt\small murtis@umn.edu\}} are with Department of Electrical and Computer Engineering, University of Minnesota, Minneapolis, USA,
}
}
\begin{document}

\maketitle

\begin{abstract}
Motivated by the needs of resiliency, scalability and plug-and-play operation, distributed decision making is becoming increasingly prevalent. The problem of achieving consensus in a multi-agent system is at the core of distributed decision making. In this article, we study the problem of achieving average consensus over a \textit{directed} multi-agent network when the communication links are corrupted with \textit{noise}. We propose an algorithm where each agent updates its estimates based on local mixing of information and adds its weighted noise-free initial information to its updates during every iteration. We demonstrate that with appropriately designed weights the agents achieve consensus under additive communication noise. We establish that when the communication links are \textit{noiseless} the proposed algorithm moves towards consensus at a geometric rate. Under communication noise we prove that the agent estimates reach to a consensus value \textit{almost surely}. We present numerical experiments to corroborate the efficacy of the proposed algorithm under different noise realizations and various algorithm parameters.
\\\\
\textit{keywords}: Average consensus, push sum, ratio consensus, communication noise,  directed graphs, multi-agent networks, almost sure convergence.

\end{abstract}

\begin{section}{Introduction}\label{sec:introduction}
Enabled by the large number of low-cost sensors, distributed computation and storage, multi-agent systems (MAS), are being leveraged across diverse applications in modern day systems including, power systems \cite{patel2017distributed, patel2020distributed, patel2020codit}, banking and financial systems \cite{ala2016classifiers}, unmanned aerial and autonomous underwater vehicle control \cite{FaxMur04, olfati2007consensus}, and machine learning \cite{nedic2020distributed} that contain multiple decision makers. Due to the inherent distributed nature of the available information in MAS, a distributed form of decision making is required \cite{arrow1958decentralization, degroot1974reaching, lynch1996distributed}. Moreover, distributed control applications need estimates of the global quantities to be made available locally to design inputs of the agents \cite{jadbabaie2003coordination, olfati2007consensus}. A consensus algorithm allowing all agents to agree upon the final decision (or state value) is an efficient way of getting global information/decisions in MAS. The problem of reaching to a consensus in MAS is studied widely in systems and control theory (see \cite{olfati2004consensus, kempe2003gossip, hadjicostis2012average,  saraswat2019distributed, hadjicostis2013average, mangalJournal, melbourne2020geometry,   melbourne2020convex,   olfati2007consensus}
and references therein). A special case of consensus among the agents is when the agreed upon state is the average of all the initial values of the agents; called average consensus. 


A key issue in the application of consensus protocols in MAS arises due to the noise in wireless communication channels between the agents.  The noisy communication links can obscure the true states of the agents and can drive the agent estimates far away from the original consensus value (the average) or even lead to divergence of the consensus protocol. Researchers have investigated the problem of achieving consensus under communication noise in the literature \cite{liu2011distributed, xiao2007distributed, carli2009average, pan2015consensus, hanada2020stochastic, dasarathan2015robust, li2018analysis, huang2007stochastic, kar2007distributed, kar2008distributed, pescosolido2008average, li2010consensus, aysal2010convergence, rajagopal2010network, schizas2007consensus, schizas2008consensus, huang2009coordination, kibangou2011finite, zhou2013discrete, he_bounded, rego2015consensus, jadbabaie2016performance, morral2017success, chen2017critical, granichin2020simultaneous, mateos2016noise, touri2009distributed, wang2010dynamic, wang2013distributed, pu2018flocking, long2015distributed, sheipak2020reaching, huang2009stochastic, wang2015consensus, li2009mean, yaziciouglu2020high, wang2009distributed, wang2016robust, zong2015stochastic, cheng2013mean}. Majority of these works devise consensus protocols with discrete-time updates \cite{liu2011distributed, xiao2007distributed, carli2009average, pan2015consensus, hanada2020stochastic, dasarathan2015robust, li2018analysis, huang2007stochastic, kar2007distributed, kar2008distributed, pescosolido2008average, li2010consensus, aysal2010convergence, rajagopal2010network, schizas2007consensus, schizas2008consensus, huang2009coordination, kibangou2011finite, zhou2013discrete, he_bounded, rego2015consensus, jadbabaie2016performance, morral2017success, chen2017critical, granichin2020simultaneous, mateos2016noise, touri2009distributed, wang2010dynamic, wang2013distributed, pu2018flocking, long2015distributed, sheipak2020reaching, huang2009stochastic, wang2015consensus}; while some researchers have presented consensus algorithms in the continuous time setting \cite{li2009mean, yaziciouglu2020high, wang2009distributed, wang2016robust, zong2015stochastic, cheng2013mean}. The articles \cite{liu2011distributed,li2009mean, xiao2007distributed, carli2009average, pan2015consensus, dasarathan2015robust, li2018analysis, kar2007distributed, kar2008distributed, pescosolido2008average, li2010consensus, rajagopal2010network, schizas2007consensus, schizas2008consensus, huang2009coordination, kibangou2011finite, zhou2013discrete, he_bounded, rego2015consensus, jadbabaie2016performance, chen2017critical, granichin2020simultaneous, yaziciouglu2020high, touri2009distributed, wang2009distributed, wang2010dynamic, wang2013distributed, pu2018flocking, long2015distributed, sheipak2020reaching,  wang2015consensus, wang2016robust, zong2015stochastic, cheng2013mean} are focused on undirected graphs where the agents are connected via bidirectional edges and the communication topology is symmetric. Moreover, most of the existing algorithms are \textit{synthesized centrally} by utilizing global information of the interconnection structure. Here, a symmetric double stochastic matrix (sum of entries of all the rows and columns is equal to one) or a balanced Laplacian matrix is utilized to capture the communication interconnection among the agents \cite{liu2011distributed, li2009mean, xiao2007distributed, carli2009average, pan2015consensus, hanada2020stochastic, dasarathan2015robust, li2018analysis, kar2007distributed, kar2008distributed, pescosolido2008average, li2010consensus, rajagopal2010network, schizas2007consensus, schizas2008consensus, kibangou2011finite, zhou2013discrete, he_bounded, rego2015consensus, chen2017critical, granichin2020simultaneous, yaziciouglu2020high, mateos2016noise, touri2009distributed, wang2009distributed, wang2010dynamic, wang2013distributed, pu2018flocking, long2015distributed, sheipak2020reaching, wang2015consensus, wang2016robust, zong2015stochastic, cheng2013mean}. The consensus protocols in \cite{huang2007stochastic, huang2009coordination, huang2009stochastic} employ a row stochastic (sum of entries of all the rows is equal to one) update. However, in case of row stochastic weight matrices the estimates incur a bias away from the average of the initial values of the agents. Articles \cite{wang2009distributed, wang2010dynamic, wang2013distributed} require a balanced feasible graph Laplacian for achieving agreement between the agents. In most schemes, that address the communication channel noise, it is assumed that the noise is a zero mean random variable independent across different agents and time-instants \cite{liu2011distributed, li2009mean, xiao2007distributed, carli2009average, pan2015consensus, hanada2020stochastic, li2018analysis, huang2007stochastic,huang2009stochastic, kar2007distributed, kar2008distributed, pescosolido2008average, li2010consensus, aysal2010convergence, rajagopal2010network, schizas2007consensus, schizas2008consensus, huang2009coordination, kibangou2011finite, jadbabaie2016performance, chen2017critical, granichin2020simultaneous, yaziciouglu2020high, mateos2016noise, touri2009distributed, wang2009distributed, wang2010dynamic, wang2013distributed, pu2018flocking, sheipak2020reaching, wang2015consensus, wang2016robust, zong2015stochastic, cheng2013mean}.
Article \cite{carli2009average} considers the channel noise to have a Poisson distribution, \cite{dasarathan2015robust} assumes a zero median (e.g., its PDF is symmetric about zero) communication noise, \cite{huang2007stochastic} analyzes the algorithm under the assumption of communication noise with bounded $p$-th moments i.e., $\mathbb{E}| w_t|^p < \infty, p \in (1,2]$, and in \cite{long2015distributed} a multiplicative noise model is considered where the noise intensities depend on the relative states of agents and form a martingale. A bounded noise model is assumed in \cite{zhou2013discrete, he_bounded, rego2015consensus, granichin2020simultaneous} where, the communication link noise is allowed to change randomly between an upper and lower bound. In terms of the convergence guarantees the existing literature can be categorized into three categories: (i) articles \cite{liu2011distributed, li2009mean, xiao2007distributed, carli2009average, pan2015consensus, li2018analysis,huang2007stochastic, li2010consensus, schizas2007consensus, schizas2008consensus, huang2009coordination, kibangou2011finite, chen2017critical, touri2009distributed, wang2009distributed, wang2010dynamic, wang2013distributed, pu2018flocking, long2015distributed, huang2009stochastic, wang2015consensus, wang2016robust, zong2015stochastic, cheng2013mean} provide guarantees about weak consensus where the agreement between the agents is reached in the mean square sense; (ii) works in \cite{liu2011distributed, dasarathan2015robust,huang2007stochastic, kar2007distributed, kar2008distributed, li2010consensus, aysal2010convergence, rajagopal2010network, morral2017success, touri2009distributed, wang2009distributed, long2015distributed, huang2009stochastic, zong2015stochastic} establish a strong consensus guarantee which provide an almost sure agreement between the agent states. However, no analysis about closeness of the asymptotic consensus value to the average of the initial values of the agents is provided; (iii) under the assumption of bounded unknown noise, analysis in \cite{zhou2013discrete,he_bounded, rego2015consensus, granichin2020simultaneous, mateos2016noise, sheipak2020reaching} provide an uniform upper bound on the state residuals (difference between the estimates of different agents) of the agents in the MAS; however the agent estimates are not shown to achieve consensus.

The algorithms proposed in the literature under communication noise \cite{liu2011distributed, xiao2007distributed, carli2009average, pan2015consensus, hanada2020stochastic, dasarathan2015robust, li2018analysis, huang2007stochastic, kar2007distributed, kar2008distributed, pescosolido2008average, li2010consensus, aysal2010convergence, rajagopal2010network, schizas2007consensus, schizas2008consensus, huang2009coordination, kibangou2011finite, zhou2013discrete, he_bounded, rego2015consensus, jadbabaie2016performance, morral2017success, chen2017critical, granichin2020simultaneous, mateos2016noise, touri2009distributed, wang2010dynamic, wang2013distributed, pu2018flocking, long2015distributed, sheipak2020reaching, huang2009stochastic, wang2015consensus, li2009mean, yaziciouglu2020high, wang2009distributed, wang2016robust, zong2015stochastic, cheng2013mean} work over undirected graphs, require symmetric and balanced algorithm structures, and cannot be synthesized locally by agents without the knowledge of global information. In recent years, there is a wide interest in having distributed decision making (distributed optimization and consensus) over general directed graph topologies (see \cite{khatana2020gradient, khatana2020d, nedic2014distributed, nedic2017achieving, pu2020push} and references therein). A natural question that comes up is to develop a distributed average consensus algorithm that can work over directed graph topologies under communication noise. In the \textit{noiseless} case, article \cite{kempe2003gossip} developed an average consensus algorithm called the PushSum/Ratio-consensus algorithm with a focus on achieving average consensus in MAS with directed interconnection structure. An important property of the PushSum algorithm is that it can be \textit{synthesized locally} by each agent in the MAS without requiring any global information. The PushSum algorithm is used extensively in the distributed optimization literature for distributed decision making over general directed graphs with noiseless communication (see \cite{khatana2020gradient,khatana2019gradient, khatana2020d,khatana2020dc, nedic2014distributed, nedic2017achieving} and references therein). Articles \cite{hadjicostis2012average, saraswat2019distributed} have studied the convergence of the PushSum algorithm when the underlying communication topology of the MAS is time-varying. The properties of the PushSum algorithm under delays in communication channels is studied in \cite{hadjicostis2013average, mangalJournal}. Authors in \cite{melbourne2020geometry, melbourne2020convex, khatana2020gradient} have extended the PushSum algorithm to higher dimensions where the decision variables are vector valued. Moreover, a finite-time termination criteria for the PushSum algorithm is developed in \cite{saraswat2019distributed,mangalJournal, melbourne2020geometry, melbourne2020convex, khatana2020gradient}. Although the articles \cite{hadjicostis2012average, saraswat2019distributed, hadjicostis2013average, mangalJournal, melbourne2020geometry, melbourne2020convex, khatana2020gradient} have extended the applicability of the PushSum algorithm under relaxed graph and communication scenarios, the performance guarantees of the PushSum algorithm under communication noise remains unaddressed. To this end, the current work addresses the following key issues:
\begin{itemize}
    \item[1.] Can the celebrated PushSum algorithm that is utilized widely in the distributed decision making applications over directed graphs, work under communication noise?
    \item[2.] If the Push-Sum algorithm and its variants in the state-of-the-art do not perform satisfactorily under noisy communication, then device  a consensus algorithm that can with all the attractive features of the Push-Sum algorithm: distributed synthesis and implementation,  guarantees about the agents achieving consensus and converging to the average of the initial values of the agents.
\end{itemize}

In this article, we address the problem of designing an algorithm to achieve average consensus in directed multi-agent graphs where the communication links between the agents are corrupted with noise. In the proposed algorithm the agents update their states by local mixing among their neighbors connected via a directed link. We establish guarantees about agreement between the agents' states under the proposed scheme. Further, we corroborate the performance of the developed algorithm with numerical simulations. The contributions of this article are summarized below:
\begin{itemize}
    \item[1)] The current article establishes that the PushSum algorithm does not perform well under imperfect communication with additive noise. 
    
    \item[2)] An important contribution of the current work is the proposed \textbf{\underline{N}oise \underline{R}esilient PushSum} ($\nrps$) algorithm $\nrps$ that extends the well studied PushSum algorithm to practical settings with noisy communication channels. We emphasize that the current work is the \textit{first} work to present such an important extension.
    
    \item[3)] The developed $\nrps$ algorithm has the following features: 
    \begin{itemize}
        \item[a)] Compared to existing algorithms in the literature that focus on getting average consensus under noisy communication, the proposed $\nrps$ algorithm achieves average consensus over general directed graph topologies and does not require restrictive assumptions of the communication graph to be symmetric and/or balanced. 
        
        \item[b)] The $\nrps$ algorithm can be synthesized distributively. In particular, $\nrps$ can be {\it designed} (and not just implemented distributively) using only decentralized information. This is a significant advantage that enables any new agent to join the network with ease.
        
        \item[c)] The $\nrps$ algorithm achieves average consensus at a geometric rate under noiseless communication. In case of the communication noise, agents' estimates under the $\nrps$ algorithm achieve consensus (perfect agreement) \textit{almost surely}. Moreover, the algorithm parameters can be controlled to make the consensus value close to the average of the initial states of the agents.
    \end{itemize}

\end{itemize}

The rest of the article is organized as follows. In Section~\ref{sec:defn_probdes}, we discuss the problem description and present some basic definitions and notations that are used in the article. Section~\ref{sec:algorithm_section} presents the proposed $\nrps$ algorithm in detail. In Section~\ref{sec:convgAnalysis}, we provide the convergence analysis of the $\nrps$ algorithm under both: perfect and noisy communication. Numerical simulations are provided in Section~\ref{sec:sim_results} to demonstrate the performance of the $\nrps$ algorithm. Section~\ref{sec:conclusion} presents the concluding remarks.
\end{section}

\begin{section}{Definitions, Notations and Problem Description}\label{sec:defn_probdes}
\subsection{Definitions and notations}
First, we present some definitions that will be useful for the rest of the development. Detailed description of most of these notions are available in \cite{Die06} and \cite{horn2012matrix}.

\begin{definition}(Directed Graph)
A directed graph $\mathcal{G}$ is a pair $(\mathcal{V},\mathcal{E})$ where $\mathcal{V}$ is a set of vertices (or nodes) and $\mathcal{E}$ is a set of edges, which are ordered subsets of two distinct elements of $\mathcal{V}$. If an edge from $j \in \mathcal{V}$ to $i \in \mathcal{V}$ exists then it is denoted as $(i,j)\in \mathcal{E}$. 
\end{definition}

\begin{definition}(Path) 
In a directed graph, a directed path from node $i$ to $j$
exists if there is a sequence of distinct directed edges of $\mathcal{G}$ of
the form $(k_{1},i),(k_{2},k_{1}),...,(j,k_{m}).$
\end{definition}

\begin{definition}(Strongly Connected Graph) A directed graph is strongly connected if it has a directed path between each pair of distinct nodes $i$ and $j.$ 
\end{definition}

\begin{definition}(In-neighborhood) The set of in-neighbors of node $i \in \mathcal{V}$ is denoted by $\mathcal{N}^{I}_i = {\{j: (i,j)\in \mathcal{E}}\}$. The number of agents $|\mathcal{N}^{I}_i|$ in the in-neighborhood of an agent $i$ is called the in-degree of agent $i$.
\end{definition}

\begin{definition}(Out-neighborhood) The set of out-neighbors of node $i \in \mathcal{V}$ is denoted by $\mathcal{N}^{O}_i = {\{j: (j,i)\in \mathcal{E}}\}$. The number of agents $|\mathcal{N}^{O}_i|$ in the out-neighborhood of an agent $i$ is called the out-degree of agent $i$.
\end{definition}
\begin{definition}(Column Stochastic Matrix) A real $n\times n$ matrix $A=[a_{ij}]$
is called a column-stochastic matrix if $0 \leq a_{ij}\leq 1$ 
and $\sum_{i=1}^{n}a_{ij}=1$ for $1\leq i,j\leq n.$ 
\end{definition}

\begin{definition}(Irreducible Matrix)  An $n \times n$ matrix $A$ is reducible if we may partition
$\{1,\dots,n\}$ into two non-empty subsets $E$, $F$ such that $A_{ij} = 0$ if $i \in E, j \in F$. Matrix $A$ is called irreducible if it is not reducible. 
\end{definition}

\begin{definition}(Aperiodic Matrix)
Let $A$ be a $n \times n$ matrix. Let $\G(\V,\E)$ be a directed graph induced by the zero/non-zero structure of $A$. Then matrix $A$ is aperiodic if $\G(\V,\E)$ is strongly connected with at least one node has a self loop.
\end{definition}

\begin{definition}(Filtration)
Let $(\Omega ,{\mathcal {A}},P) $ be a probability space and let $T$ be an index set with a total order $\leq$  (often $ \mathbb {N} , \mathbb {R}^{+},$ or a subset of $\mathbb {R}^{+}$). For every $ t\in T$ let $\mathcal {F}_{t}$ be a sub-$\sigma$-algebra of $\mathcal {A}$. Then the collection $\{\mathcal {F}_{t}\}_{t \in T}$ is called a filtration, if $\mathcal {F}_{k}\subseteq \mathcal {F}_{\ell}$ for all $ k\leq \ell$. Thus, filtrations are families of $\sigma$-algebras that are ordered non-decreasingly.
\end{definition}

Define $[n] := \{1,2,\dots,n\}$. We denote by $p_{ij}$ the entry in $i^{th} \in [m]$ row and $j^{th} \in [n]$ column of the matrix $P \in \mathbb{R}^{m\times n}$. $\mathbf{1}$ denotes a vector of appropriate dimension with all entries equal to one. $\I_q \in \mathbb{R}^{q \times q}$ denotes the identity matrix of dimension $q$. For a vector $b \in \mathbb{R}^q$, we denote $\text{\textbf{diag}}(b)$, to be a $q \times q$ square diagonal matrix with the elements of vector $b$ on the main diagonal. 

We denote the Hadamard (or element-wise/entry-wise) product and division of matrices by $\odot$ and $\oslash$ respectively. $O(.)$ and $o(.)$ denote the standard \textit{Big-O} and \textit{Little-o} notations respectively. We use the term \textit{almost surely} to state that an event happens with probability 1. We use the words node and agent interchangeably in the text.

\subsection{Setup}
Consider a set of $n$ agents connected through a directed graph $\mathcal{G}(\mathcal{V},\mathcal{E})$. The agents communicate via uni-directional communication links in $\mathcal{G}(\mathcal{V},\mathcal{E})$. The communication link between any agent $i$ and its neighbor $j$ is corrupted by noise. Each agent $i \in \mathcal{V}$ has an initial state $u_i \in \mathbb{R}^p$. The objective is to design an algorithm that allows each agent to distributively compute the average, 
\begin{align}\label{eq:ini_avg}
    \overline{u}:= \textstyle \frac{1}{n} \sum_{i=1}^n u_i,
\end{align}
by interacting only with its neighboring agents in $\mathcal{G}(\mathcal{V},\mathcal{E})$ under the communication constraints: uni-directional communication links with additive noise.

\subsection{PushSum with noisy communication}
Under the PushSum algorithm, at any discrete time-instant $k$ each agent $i$ maintains three estimates $x_i(k) \in \mathbb{R}^p, y_i(k) \in \mathbb{R}$ and $z_i(k) \in \mathbb{R}^p$. Due to noise in communication channels in $\G(\mathcal{V},\mathcal{E})$ the agent $i$ receives noisy version of the information sent by the agent $j$.
Node $i$ updates its estimates at $(k+1)^{th}$ discrete time-instant, by performing a weighted combination of the (noisy) information received from its in-neighbors and its own information (available without any additive noise at the agent $i$) as follows:
\begin{align}
 x_{i}(k+1) & = p_{ii}x_{i}(k) + \textstyle \sum_{j \in \mathcal{N}^{I}_i }(p_{ij} x_{j}(k) + \eta_{x_{ij}}(k)), \label{eq:consensus_num_ps}\\
y_{i}(k+1) & = p_{ii}y_{i}(k) + \textstyle \sum_{j \in \mathcal{N}^{I}_i }(p_{ij} y_{j}(k) + \eta_{y_{ij}}(k)), \label{eq:consensus_den_ps}
\end{align}
where, $ x_i(0) = u_i$, and  $y_i(0) = 1$ for all $i \in [n]$ and the terms $\eta_{x_{ij}}(k)$ and,  $\eta_{y_{ij}}(k)$ are the additive noise realizations in the estimates of agent $j$ sent to agent $i$ at any instant $k$. The variable $z_i(k)$ is the estimate of the average  $\overline{u}$ maintained by agent $i \in \mathcal{V}$ at time-instant $k$ which is updated as,
\begin{align}
z_i(k+1) &=  \textstyle \frac{1}{y_i(k+1)}x_i(k+1). \label{eq:consensus_ratio_ps}
\end{align}

\noindent Let $\p$ be weight matrix corresponding to~(\ref{eq:consensus_num_ps}) and~(\ref{eq:consensus_den_ps}) with $\p_{ij} \in [0,1]$ and,
\begin{align}
    \p_{ij} &= \begin{cases}\label{eq:pmat}
    p_{ij} \neq 0, & \text{if} \ i = j \ \text{or} \ (i,j) \in \E, \\
    0, &  \text{otherwise}. 
\end{cases}
\end{align}

\noindent The following assumption is typically made to study the PushSum algorithm:
\begin{assump}\label{assp:strg_colmstoc}
The directed graph $\mathcal{G}(\mathcal{V},\mathcal{E})$ is strongly connected and associated weight matrix $\p$ is column stochastic. 
\end{assump}

\noindent One choice of weights that satisfy Assumption~\ref{assp:strg_colmstoc} is the \textit{out-degree based equal neighbor weights rule}: Here, 
\begin{align}
    p_{ij} &= \begin{cases}\label{eq:weightrule}
    \textstyle \frac{1}{1+|\mathcal{N}^{O}_j|}, & \text{if} \ i \in \mathcal{N}^{O}_j, \ \forall j, \\
    0, &  \text{if}\ i \notin \mathcal{N}^{O}_j, \ \forall j. 
\end{cases}
\end{align}
Define for all $i \in [n]$,
\begin{align}\label{eq:notation_1}
 \eta_{x_i}(k) := \textstyle \sum_{j \in \mathcal{N}^{I}_i } \eta_{x_{ij}}(k), \ \eta_{y_i}(k) := \textstyle \sum_{j \in \mathcal{N}^{I}_i } \eta_{y_{ij}}(k).
\end{align}
We stack the agents' estimates in a column vector form and define the following quantities:
\begin{align}
\x(k) & := [x_1(k); \dots; x_n(k)] \in \mathbb{R}^{n \times p}, \label{eq:notation_2} \\ 
\y(k) & := [y_1(k); \dots; y_n(k)] \in \mathbb{R}^{n}, \label{eq:notation_3}\\
 \eta_{x}(k) & := [\eta_{x_1}(k); \dots; \eta_{x_n}(k)] \in \mathbb{R}^{n \times p}, \ \text{and} \label{eq:notation_4} \\ 
 \eta_{y}(k) &:= [\eta_{y_1}(k); \dots; \eta_{y_n}(k)] \in \mathbb{R}^{n}.\label{eq:notation_5}
\end{align}
Thus, updates~(\ref{eq:consensus_num_ps}) and~(\ref{eq:consensus_den_ps}) can be written as:
\begin{align} \nonumber
\hspace{-0.1in} \x(k+1) = \p \x(k) + \eta_{x}(k), 
\y(k+1) = \p \y(k) + \eta_{y}(k),
\end{align}
where, $\eta_x(k)$ and $\eta_y(k)$ are the additive noise. Left multiplying by $\mathbf{1}^\top$ on both sides of the above equations, we have 
\begin{align*}
    \mathbf{1}^\top \x(k+1) &= \mathbf{1}^\top \x(k) + \mathbf{1}^\top \eta_x(k), \ \text{and} \\
    \mathbf{1}^\top \y(k+1) &= \mathbf{1}^\top \y(k) + \mathbf{1}^\top \eta_y(k).
\end{align*}
Thus, the sum of estimates $\x(k)$ and $\y(k)$ integrate the communication noise and thus is equivalent to a random walk with unbounded variance. Note, that in the case when the communication links are noiseless the sum of the estimates $\x(k)$ and $\y(k)$ are conserved during the iterations \cite{patel2020distributed} that plays a crucial role in the convergence of the PushSum algorithm to the average of the initial values. The random walk behavior of the sums imply that the estimates $\x(k)$ and $\y(k)$ and hence, $\z(k)$ do not converge to a meaningful value. This is clear from the trajectories of the agent states $\z_i(k)$ under the PushSum algorithm for a $10$ node directed graph with noisy communication shown in Fig.~\ref{fig:ratio}. Here, the initial values of the agents are chosen to be $[1,2,3,4,5,6,7,8,9,10]$ (with the average being $5.5$). The communication links have an additive noise modeled as a normal random variable, $\mathcal{N}(0,0.1)$. 
\begin{figure}[h]
\centering
    \includegraphics[scale=0.21,trim={0.35cm 5.2cm 1.12cm 6.1cm},clip] {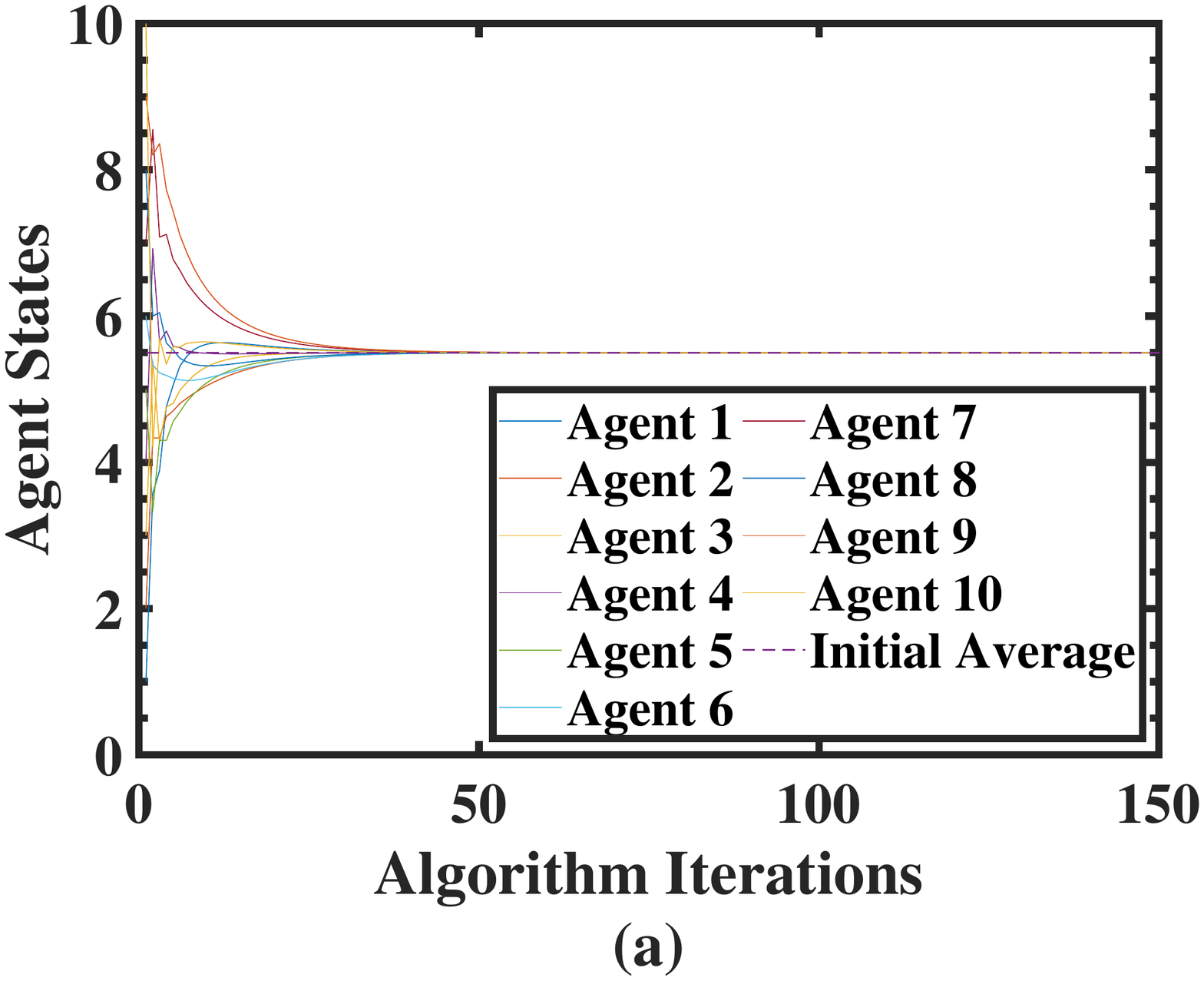}
    \includegraphics[scale=0.21,trim={0.35cm 5.2cm 1.12cm 6.1cm},clip] {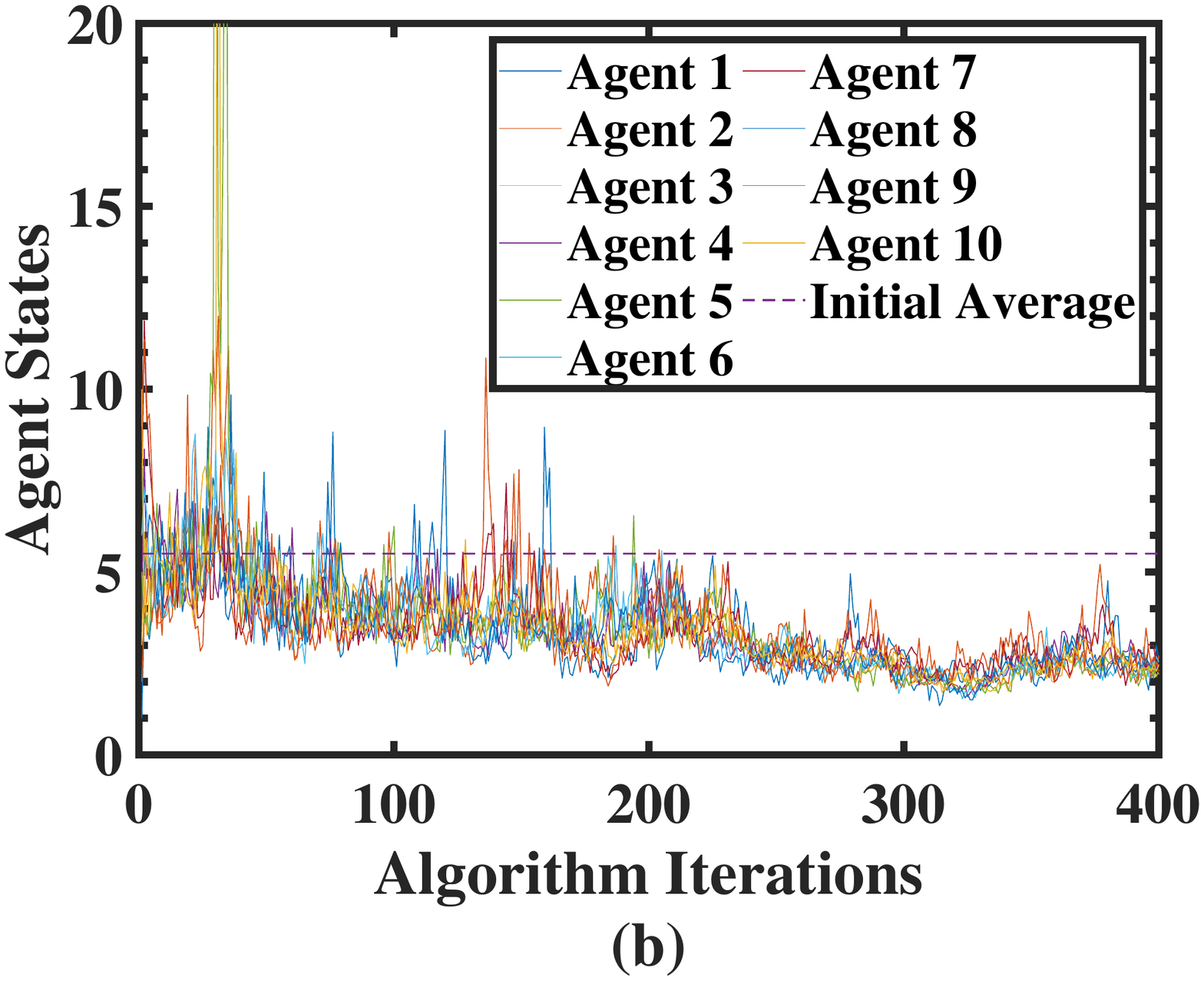}
  \caption{Performance of the PushSum algorithm under the perfect communication (a) and with communication noise (b).}
  \label{fig:ratio} 
\end{figure}
These observations suggest that the PushSum algorithm \textit{does not} perform adequately under communication noise. 
\end{section}

\section{Distributed Average Consensus on noisy directed networks}\label{sec:algorithm_section}

In this section, we propose a distributed algorithm which allows all the agents to determine the initial average $\overline{u}$ (defined in~(\ref{eq:ini_avg})) under noisy communication. We call it the $\nrps$ algorithm as it is based on the PushSum algorithm, similar to which here each agent $i$, at any discrete time-instant $k$, maintains three estimates $x_i(k) \in \mathbb{R}^p, y_i(k) \in \mathbb{R}$ and $z_i(k) \in \mathbb{R}^p$. We utilize two strategies to design the $\nrps$ algorithm that counter the effect of the communication noise: \\
1. The estimates $x_i(k), y_i(k)$ are updated via a local mixing across the neighborhood of the agent $i \in \mathcal{V}$. Every agent $i$, weighs the information sent by its neighbors (agent $j \in \mathcal{N}^{I}_i $) and its own information to appropriately control the amount of the noise to be added to its updated estimates.  \\
2. Each agent $i$, utilizes its initial values $x_i(0)$ and $y_i(0)$ to update its estimates. The initial values $x_i(0)$ and $y_i(0)$ are deterministic and known to the agent $i$. Utilizing the noise-free initial values provide a tether to the agents' estimates in tracking the initial average in presence of the communication noise. 

Incorporating these two strategies together, the estimates in the $\nrps$ algorithm are updated as follows, for $i \in \V$:
\begin{align}
 x_{i}(k+1) & = \alpha_i(k) p_{ii}x_{i}(k) + \beta(k) \textstyle \sum_{j \in \mathcal{N}^{I}_i }(p_{ij} x_{j}(k) \nonumber \\  & \hspace{0.98in} + \eta_{x_{ij}}(k)) + \theta(k) x_i(0), \label{eq:consensus_num}\\
y_{i}(k+1) & = \alpha_i(k) p_{ii}y_{i}(k) + \beta(k) \textstyle \sum_{j \in \mathcal{N}^{I}_i }(p_{ij} y_{j}(k) \nonumber \\  & \hspace{0.98in} + \eta_{y_{ij}}(k)) + \theta(k) y_i(0), \label{eq:consensus_den}
\end{align}
\begin{figure}[!b] 
\centering
    \includegraphics[scale=0.28,trim={1.7cm 0cm 2cm 2.5cm},clip] {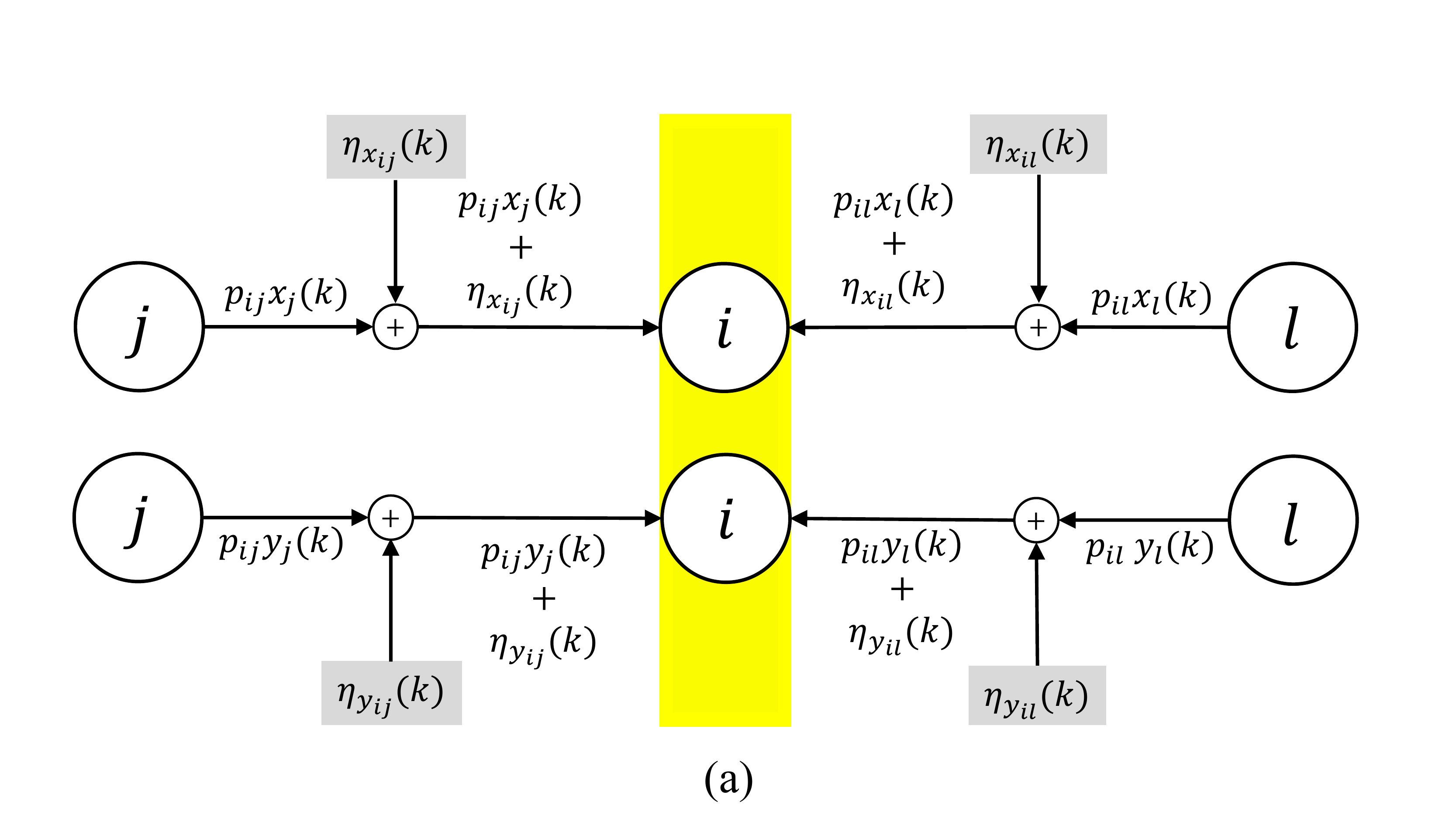}
    \includegraphics[scale=0.28,trim={4.8cm 0cm 1.2cm 0cm},clip] {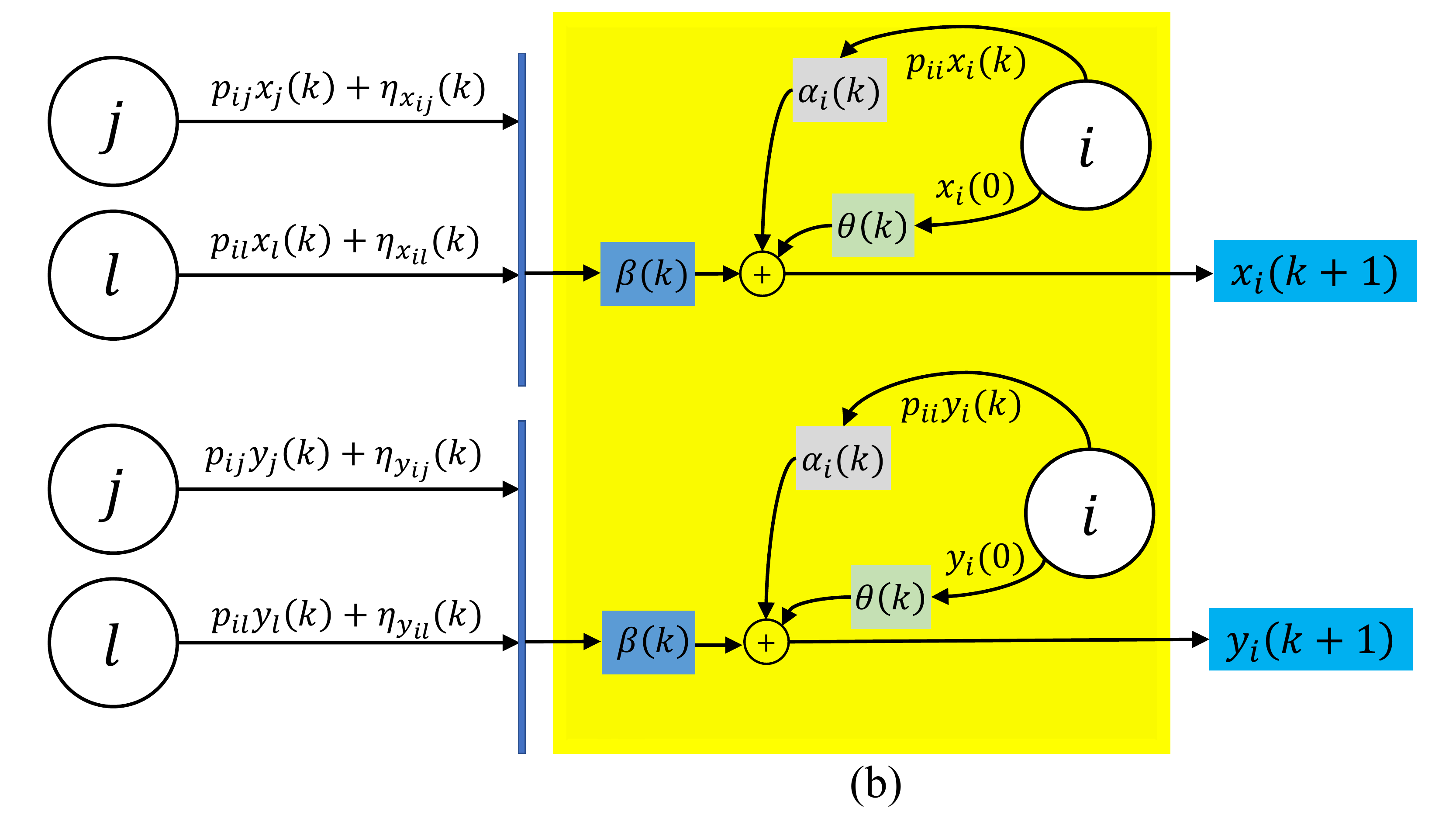}
  \caption{Estimate updates of the $\nrps$ algorithm: (a) agent $i$ receives noisy information from neighbors $j$ and $l$, (b) agent $i$ updates its estimates using the received information, its own current estimate and its initial values.}
  \label{fig:nrps} 
\end{figure}
respectively, where, $ x_i(0) = u_i$, and  $y_i(0) = 1$ for all $i \in [n]$. $\alpha_i (k) \in [0,1], \forall i \in [n], k,$ and $\beta(k) \in [0,1], \forall k$ are the relative weights that agent $i$ utilize to control the information mixing, and $\theta(k) \in \mathbb{R}, \forall k$, is utilized to control the influence of the noise-free initial values on the estimate updates. The variable $z_i(k)$ maintained by agent $i \in \mathcal{V}$ at time-instant $k$ is updated as in the PushSum algorithm described by
\begin{align}
z_i(k+1) &=  \textstyle \frac{1}{y_i(k+1)}x_i(k+1). \label{eq:consensus_ratio}
\end{align}
Fig.~\ref{fig:nrps} shows the estimate updates for an agent $i \in \V$ at any discrete-time instant $k$. Fig.~\ref{fig:nrps}(a) shows the communication path where agent $i$ receives information from its in-neighbors $j$ and $l$. Fig.~\ref{fig:nrps}(b) shows the computation step where agent $i$ updates its estimates based on the information from the neighbors $j$ and $l$, its own current state values $(x_i(k),y_i(k))$ and the noise-free initial values $(x_i(0),y_i(0))$. 
Define $\widehat{\p}$ such that
\begin{align}
  \widehat{\p}_{ij} &= \begin{cases}\label{eq:phat}
    \p_{ij}, & \text{if} \ i \neq j \\
    0, &  \text{if} \ i = j. 
\end{cases}
\end{align}

\noindent Note that due to Assumption~\ref{assp:strg_colmstoc} the matrix $\widehat{\p}$ defined in~(\ref{eq:phat}) is a column sub-stochastic matrix, i.e., $\sum_{i=1}^n \widehat{\p}_{ij} < 1 $ for all $j \in [n]$. The parameter $\alpha_i(k), i \in [n]$ is set as:
\begin{align}\label{eq:alpha}
    \alpha_i(k) = \textstyle \frac{1 - \beta(k)(1 - p_{ii})}{p_{ii}}.
\end{align}

\noindent Let $\A(k) := (\textbf{diag}(\alpha_i(k)) \odot \p)$ and $\B(k) := \beta(k) \widehat{\p}$. Using the defined notations in~(\ref{eq:notation_1})-(\ref{eq:notation_5}), updates~(\ref{eq:consensus_num})-(\ref{eq:consensus_ratio}) can be written as:
\begin{align}
\x(k+1) & =  \A(k) \x(k) + \B(k) \x(k) + \theta(k) \x(0) \nonumber \\ & \hspace{1.5in} + \beta(k) \eta_{x}(k), \label{eq:compact_num}\\
\y(k+1) & = \A(k) \y(k) + \B(k) \y(k) + \theta(k) \y(0) \nonumber \\
& \hspace{1.5in} + \beta(k) \eta_{y}(k), \label{eq:compact_den} \\
\z(k+1) & = \x(k+1) \oslash \y(k+1). \label{eq:compact_ratio}
\end{align}

\noindent The initial values of the agents in updates~(\ref{eq:consensus_num})-(\ref{eq:consensus_ratio}) help preserve the true information of the agents and reduce the effect of the communication noise. Moreover, the agents vary the weights $\beta(k)$ and $\theta(k)$  dynamically that allows the agents' estimates $x_i(k)$ and $y_i(k)$ reach a steady state. Specifically, the weights $\beta(k)$ and $\theta(k)$ satisfy the following assumptions:
\begin{assump}\label{assp:beta}
The sequence $\{\beta(k)\}_{k\geq 0}$ is such that:
\begin{itemize}
    \item[i)]  $\beta(k) \in [0, 1)$, for all $k$,
    \item[ii)] there exists a $K_\beta \geq 1$ such that
      \begin{align}\label{eq:beta}
          a k^{-q} \leq \beta(k) \leq b k^{-q},
      \end{align}
     for all $k \geq K_\beta$, where $q > 1$ and $0 \leq a \leq b < \infty$. 
\end{itemize}
\end{assump}
\noindent Note that the requirement~(\ref{eq:beta}) with a suitable $K_\beta$, allows more flexibility in choosing $\beta(k)$ as larger values of $a$ and $b$ can be chosen while maintaining $\beta(k) \in [0,1]$, $k \geq K_\beta$. In the sequel, the parameters $a,b$ and $q$ will be treated as fixed constants.

\begin{assump}\label{assp:theta}
The sequence $\{\theta(k)\}_{k \geq 0}$ satisfies:
\begin{align}\label{eq:theta}
 \theta(k) \in [0, \infty), \ \text{for all} \ k, \ \text{and} \ \sum_{k = 0}^{\infty} \theta(k) < \infty.
\end{align}
\end{assump}
\noindent Note that one can appropriately choose a $K_\theta > 0$ with large values of $\theta(k)$ chosen early on, i.e. when $k < K_\theta$, to give more importance to the initial values of the agents and then $\theta(k)$ be gradually reduced, i.e. $\theta(k) \in [0,1]$ for $k > K_{\theta}$ to satisfy Assumption~\ref{assp:theta}. In section~\ref{sec:sim_results}, we demonstrate the performance of the proposed $\nrps$ algorithm under different choices of the $\beta(k)$ and $\theta(k)$ sequences.

\begin{remark}
The proposed $\nrps$ algorithm inherits nice properties of the PushSum algorithm, i.e. the $\nrps$ algorithm is applicable on general directed graph topologies and can be both implemented and synthesized/designed distributively. Recall Fig.~\ref{fig:nrps}(a), every agent $j \in \V$ decide the weights $p_{ij}$ for sending its information over the out-going links to the neighbor agent $i$, independent of the other agents. The agent $i \in \V$ upon receiving the information from its in-neighbors adds up its own estimate and the initial value with the received information after weighing them by weights $\alpha_i(k), \theta(k)$ and $\beta(k)$ respectively. Note that the agents does not need to have any centralized information for weights $\beta(k)$ and $\theta(k)$. Each agent can dynamically update the weights $\beta(k)$ and $\theta(k)$ based on assumptions~\ref{assp:beta} and~\ref{assp:theta}.
\end{remark}

\begin{section}{Convergence Analysis}\label{sec:convgAnalysis}
We start this section by analyzing the convergence of the consensus protocol~(\ref{eq:compact_num})-(\ref{eq:compact_ratio}) in the case where the communication links are noiseless. Then, we extend our analysis to investigate the \textit{almost sure} convergence of the proposed $\nrps$ algorithm to a consensus value under the additive noise in the communication channels.

\subsection{Analysis under noiseless communication channels}
Here, we will analyze a special case of the scheme~(\ref{eq:consensus_num})-(\ref{eq:consensus_ratio}) where the communication channels are assumed to be noiseless. In particular, $\eta_{x_{ij}}(k) = 0 = \eta_{y_{ij}}(k)$ for all $i, j \in [n]$ at all time-instants $k \geq 0$.

\begin{lemma}\label{lem:ratio_well_def}
Let Assumption~\ref{assp:strg_colmstoc} hold. Assume that the communication channels in $\G(\V, \E)$ are \textbf{noiseless}. Let the sequences $\{\beta(k)\}_{k \geq 0}$ and $\{\theta(k)\}_{k \geq 0}$ satisfy assumptions~\ref{assp:beta} and~\ref{assp:theta} respectively. Under the initialization $y_i(0) = 1$ for all $i \in [n]$, there exist $y_o > 0$ such that we have $y_i(k) > y_o$ for all $k \geq 1, i \in [n]$. As a consequence the states $z_i(k)$ are well-defined for all $k \geq 1, i \in [n]$.
\end{lemma}
\begin{proof}
For $\eta_y(k) = 0$, the update $\y(k)$ in the iteration~(\ref{eq:compact_den}) becomes:
\begin{align}\label{eq:proof_y_update}
    \y(k+1) = \underbrace{(\A(k) + \B(k))}_{:= \M(k)}  \y(k)  +  \theta(k) \y(0).
\end{align}

\noindent Define the following operation, 
\begin{align}\label{eq:hatprod}
    \widehat{\prod}_{r = s}^l \M(r) := \begin{cases}
    \prod_{r=s}^l \M(r), \ \text{if} \ l \geq s \\
    \I_{n}, \ \text{otherwise}.
    \end{cases}
\end{align}
Recursive application of~(\ref{eq:proof_y_update}) using the notation in~(\ref{eq:hatprod}) gives,
\begin{align} \label{eq:y_equation_in_closed_form}
\y(k+1) &= \displaystyle \widehat{\prod}_{r=0}^k \M(r) \y(0) + \sum_{s=0}^{k} \widehat{\prod}_{r = s+1}^k \M(r) \theta(s) \y(0).
\end{align}

\noindent Further define, \begin{align}\label{eq:Wk}
    W(l,k) := \textstyle \widehat{\prod}_{r=l}^k \M(r). 
\end{align}
We begin by establishing that the following relation holds.

\textit{Claim}: $\sum_{j=1}^n W_{ij}(l,k) > w_o^{k-l+1}$, for all $k \geq l \geq 1$, where, $w_o > 0$ is a constant.\\
\textit{Proof}: 
For $\M(l), l \geq 1$ note that,
\begin{align}\label{eq:bound_on_m_row_sum}
    \textstyle \sum_{j=1}^n \M_{ij}(l) &= 1 - \beta(l)(1-p_{ii}) + \beta(l) \textstyle \sum_{j \neq i}^n p_{ij} \nonumber \\
    & = \textstyle 1 - \beta(l) + \beta(l) \sum_{j=1}^n p_{ij} \nonumber \\
    & > 1 - \beta(l) > 1 - \max_{l} \beta(l) = w_o > 0.
\end{align}
We use an induction based argument for the matrices of the form $W(l, l+t)$. Let $T = k-l$. For $t = 1$, note that $W(l,l+1) = \M(l) \M(l+1)$ and the $i$-$j^{th}$ entry of the matrix $W(l,l+1)$ is given by $W_{ij}(i,+1) = [\M(l)\M(l+1)]_{ij} = \sum_{s=1}^n \M_{is}(l) \M_{sj}(l+1)$. Taking the sum over $j$ we have,
\begin{align*}
    \textstyle \sum_{j=1}^n W_{ij}(l,l+1) &= \textstyle \sum_{j=1}^n \sum_{s=1}^n \M_{is}(l) \M_{sj}(l+1) \\
    & = \textstyle \sum_{s=1}^n \left[\M_{is}(l) \sum_{j=1}^n \M_{sj}(l+1)\right]\\
    & > \textstyle \sum_{s=1}^n \M_{is}(l) w_o \\
    & > w_o^2,
\end{align*}
where, we used~(\ref{eq:bound_on_m_row_sum}) in the last two inequalities. Assume, for $t = T = k-l$, that $\sum_{j=1}^n W_{ij}(l,l+T) = \sum_{j=1}^n W_{ij}(l,k) > w_o^{k-l+1}$. For $t = T+1$, 
\begin{align*}
     \sum_{j=1}^n W_{ij}(l,l+T+1) &=  \sum_{j=1}^n \sum_{s=1}^n W_{is}(l,l+T) \M_{sj}(l+T+1) \\
    & \hspace{-0.45in} = \textstyle \sum_{s=1}^n \left[W_{is}(l,l+T) \sum_{j=1}^n \M_{sj}(l+T+1)\right]\\
    & \hspace{-0.45in} > \textstyle \sum_{s=1}^n W_{is}(l,l+T) w_o \\
    & \hspace{-0.45in} > w_o^{k-l+1}w_o = w_o^{k-l+2},
\end{align*}
where, we used~(\ref{eq:bound_on_m_row_sum}) and the induction hypothesis in the last inequality. Therefore, induction holds. Thus, 
\begin{align}\label{eq:y_sum_bound}
   \textstyle \sum_{j=1}^n W_{ij}(l,k) > w_o^{k-l+1},
\end{align}
for all $k \geq l \geq 1$. \qed 

From the definition of matrices $\A(k)$ and $\B(k)$ under Assumption~\ref{assp:strg_colmstoc} it is easy to check that the matrix $\M(k)$ is column-stochastic for all $k \geq 0$ (see Appendix~\ref{sec:W_k_cSIA}). Further, note that $W(l,k)$ for all $k, l \geq 0$ is a column-Stochastic, Irreducible and Aperiodic (cSIA) matrix (see Appendix~\ref{sec:W_k_cSIA}). Column stochasticity follows by induction and the column stochasticity of $\M(k)$. From the expression for $W(l,k)$ in~(\ref{eq:Wk}), the zero/non-zero structure of $W(l,k)$ corresponds to a graph of $n$ nodes that includes all the edges in $\G(\V,\E)$ and since the original graph $\G(\V,\E)$ is strongly connected, $W(l,k)$ corresponds to a graph that is strongly connected and hence is irreducible \cite{brualdi1991combinatorial}. Since the graph corresponding to $W(l,k)$ is strongly connected and the diagonal entries of $W(l,k)$ are positive, therefore, $W(l,k)$ is aperiodic. 
We utilize the following result from \cite{Wolfowitz} on the product of column stochastic matrices:
\begin{theorem}\label{thm:wolfowitz}
Let $\mathcal{M}: = \{\M(1), \M(2), \dots\}$ be a collection of $n \times n$ column stochastic matrices. Define a word in $\mathcal{M}$ of length $\ell$ as the product of $\ell$ matrices from the collection $\mathcal{M}$ (repetitions permitted). Let $\mathcal{M}$ be such that any word in $\mathcal{M}$ is cSIA. Then given $\varepsilon > 0$ there exists an
integer $\ell_\varepsilon$ such that any word $\mathbf{N} \in \mathbb{R}^{n \times n}$ in $\mathcal{M}$ of length $\ell \geq \ell_\varepsilon$ satisfies $\max_j \max_{i_1,i_2} |\mathbf{N}_{j,i_1} - \mathbf{N}_{j,i_2}| < \varepsilon$.
\end{theorem}
Theorem~\ref{thm:wolfowitz} implies that any sufficiently long word in the $\mathcal{M}$
has all its columns such that the difference in the columns element-wise can be made arbitrarily small. Thus, given $\varepsilon > 0$, there exists a $\ell_\varepsilon > 0$ such that $W(l,k)$ (a word in $\mathcal{M}$ that is cSIA) is equal to  $v(k)\mathbf{1}^\top + e(k)$, for all $k > l + \ell_\varepsilon$, where $v(k)$ is an appropriate vector with positive entries, and
$e(k)$ is an error matrix such that the absolute value of its elements is smaller than $\varepsilon/2$ i.e., $|e_{ij}(k)| < \varepsilon/2$ for all $ i,j \in [n]$. 

Further, recall $\M(k) := \A(k) + \B(k)$. Based on the definition of $\A(k)$ and $\B(k)$, we can re-write $\M(k)$ as
\begin{align}\label{eq:Mk_redefine}
    \M(k):= \I_n + \beta(k) \tilde{\p},
\end{align}
where, $\tilde{\p} = \textbf{diag}(p_{ii}-1) + \widehat{\p}$. 

Given any $\varepsilon > 0$ for $k > \ell_\varepsilon$, note that $W(0,k) = v(k)\mathbf{1}^\top + e(k)$, with $|e_{ij}(k)| < \varepsilon/2$ for all $i,j$.
\begin{align}\label{eq:Wk_for_proof}
    W(0,k) &= \textstyle W(0,\ell_\varepsilon)\prod_{r=\ell_\varepsilon + 1}^k \M(r) \nonumber \\
    & = \textstyle [v(\ell_\varepsilon)\mathbf{1}^\top + e(\ell_\varepsilon)] \prod_{r=\ell_\varepsilon + 1}^k \M(r) \nonumber \\
    & = \textstyle [v(\ell_\varepsilon)\mathbf{1}^\top + e(\ell_\varepsilon)] \prod_{r=\ell_\varepsilon + 1}^k (\I_n + \beta(r) \tilde{\p}) \nonumber \\
    & = \textstyle v(\ell_\varepsilon)\mathbf{1}^\top \prod_{r=\ell_\varepsilon + 1}^k (\I_n + \beta(r) \tilde{\p}) \nonumber \\
    & \textstyle \hspace{0.5in} 
    + e(\ell_\varepsilon) \prod_{r=\ell_\varepsilon + 1}^k (\I_n + \beta(r) \tilde{\p}) \nonumber \\
    & = \textstyle v(\ell_\varepsilon)\mathbf{1}^\top + e(\ell_\varepsilon) W(\ell_\varepsilon + 1, k),
\end{align}
where, we used the fact that $v(\ell_\varepsilon)\mathbf{1}^\top \tilde{\p} = \mathbf{0}_{n\times n}$.

From~(\ref{eq:y_equation_in_closed_form}),~(\ref{eq:Wk})~(\ref{eq:Mk_redefine}),~(\ref{eq:Wk_for_proof}), and initialization $y_j(0) = 1$ for all $j$ we have, 
\begin{align*}
 y_i(k+1) &= \textstyle \sum_{j=1}^n W_{ij}(0,k) + \sum_{j=1}^n \sum_{s=0}^k \theta(s) W_{ij}(s+1,k) \\
 & > \textstyle  \sum_{j=1}^n W_{ij}(0,k) \\
 & = \textstyle  \sum_{j=1}^n \big[v(\ell_\varepsilon)\mathbf{1}^\top + e(\ell_\varepsilon) W(\ell_\varepsilon + 1, k) \big]_{ij} \\
 & > \textstyle  \sum_{j=1}^n \big[v(\ell_\varepsilon)\mathbf{1}^\top - \frac{\varepsilon}{2} \mathbf{1}\mathbf{1}^\top W(\ell_\varepsilon + 1, k) \big]_{ij} \\
 & > \textstyle  \sum_{j=1}^n \big[v(\ell_\varepsilon)\mathbf{1}^\top - \frac{\varepsilon}{2} \mathbf{1}\mathbf{1}^\top \big]_{ij} \\
 & = \textstyle  n v_i(\ell_\varepsilon) - n\frac{\varepsilon}{2} \\
 & >   n \min_{1 \leq i \leq n} v_i(\ell_\varepsilon) - \textstyle \varepsilon\frac{n}{2} > 0, 
\end{align*}
where, we used the fact that $W(l,k)$ is column stochastic for all $k,l \geq 0$ in the third last inequality.
For $k \leq \ell_\varepsilon$, from~(\ref{eq:y_sum_bound}) we have
\begin{align*}
 y_i(k) & > \textstyle w_o^{k-1} + \sum_{s=0}^{k-1} \theta(s) w_o^{k-s-1}\\
 & > w_o^{k-1} + \theta(0) w_o^{k-1} \geq (1 + \theta(0)) w_o^{\ell_\varepsilon} > 0.
\end{align*}
Therefore, $y_i(k) > y_o : =\min \{n \min\limits_{1 \leq i \leq n} v_i(\ell_\varepsilon), (1 + \theta(0)) w_o^{\ell_\varepsilon}\}$, for all $k \geq 1$. Thus, $z_i(k) = \frac{1}{y_i(k)}x_i(k)$ is well defined for all $i \in [n], k \geq 0$.
\end{proof}

\begin{theorem}\label{thm:withoutnoise}
Let Assumption~\ref{assp:strg_colmstoc} hold. Assume the edges in $\G (\V, \E)$ are \textbf{not} corrupted by additive noise in communication channels. Let $\{z_i(k)\}_{k \geq 0}^{i \in [n]}$ be the sequences generated by the iteration~(\ref{eq:consensus_ratio}). Let the sequence $\{\theta(k)\}_{k \geq 0}$ satisfy Assumption~\ref{assp:theta}. Then, the estimates $z_i(k)$ converges to the average $\overline{u}$, i.e., $\lim_{k \rightarrow \infty} z_i(k) = \frac{1}{n} \sum_{j=1}^n u_j$, for all $i \in [n]$.
\end{theorem}
\begin{proof}
Similar to~(\ref{eq:proof_y_update}) the update $\x(k)$ in the iteration~(\ref{eq:compact_num}) when $\eta_x(k) = 0$ becomes:
\begin{align}\label{eq:proof_x_update}
    \x(k+1) = \M(k)  \x(k)  +  \theta(k) \x(0).
\end{align}
Recursive application of~(\ref{eq:proof_x_update}) with the notation~(\ref{eq:Wk}) gives,
\begin{align} \label{eq:equations_in_closed_form}
\x(k+1) = \underbrace{W(0,k) \x(0)}_{\text{first}} + \underbrace{\sum_{s=0}^{k} \theta(s) W(s+1,k) \x(0)}_{\text{second}}.
\end{align}
Note that as $W(l,k)$ is column stochastic for all $k,l \geq 1$, $\|W(l,k)\| \leq \overline{w} < \infty$ for all $k,l$. 

Under Assumption~\ref{assp:theta} and from Theorem~\ref{thm:wolfowitz}, given $\varepsilon > 0$. Define $\hat{\varepsilon} := \frac{\varepsilon}{\max \{ \|\x(0)\|,n \}}$. There exists $\tilde{k}$ and $v(k) \in \mathbb{R}^n$ such that $\sum_{s = \tilde{k}}^\infty \theta(s) \leq \frac{\hat{\varepsilon}}{2\overline{w}}$ and $W(l,k) = v(k)\mathbf{1}^\top + e(k)$, for all $k > \tilde{k} + l$, where, $|e_{ij}(k)| \leq \frac{\hat{\varepsilon}}{2}$ for all $i,j \in [n]$.

Thus, for all $k > \tilde{k}$, the first term in~(\ref{eq:equations_in_closed_form}) becomes,
\begin{align}\label{eq:x1}
    W(0,k)\x(0) = [v(k)\mathbf{1}^\top + e(k)]\x(0)
\end{align}
where, $|e_{ij}(k)| < \hat{\varepsilon}/2$ for all $i,j \in [n]$. Consider the second term in~(\ref{eq:equations_in_closed_form}), 
\begin{align}\label{eq:inter_proof1}
    \sum_{s=0}^{k} \theta(s) W(s+1,k) \x(0) & = \underbrace{\sum_{s=0}^{\tilde{k}} \theta(s) W(s+1,k)}_{:= \psi_1(k)}\x(0) \nonumber \\ 
    & \hspace{-0.4in} + \underbrace{\sum_{s= \tilde{k}+ 1 }^{k}  \theta(s) W(s+1,k)}_{:= \psi_2(k)}\x(0).
\end{align}
Note, $k > \tilde{k} + s + 1$ for all the terms in $\psi_1(k)$. Substituting $W(s+1,k) := v(k)\mathbf{1}^\top + \hat{e}(k)$, with $|\hat{e}_{ij}(k)| < \hat{\varepsilon}/2$ for all $i,j \in [n]$ we get,
\begin{align}\label{eq:inter_proof2}
    \textstyle \psi_1(k)\x(0) = \sum_{s=0}^{\tilde{k}} \theta(s) [v(k)\mathbf{1}^\top + \hat{e}(k)]\x(0).
\end{align}

Consider $\psi_2(k)$, 
\begin{align*}
    \psi_2(k) & = \textstyle \sum_{s=\tilde{k} + 1}^{k} W(s+1,k) \theta(s) \\ 
    & = \theta(k) + \M(k) \theta(k-1) + W(k-1,k) \theta(k-2) + \dots  \\ 
    & \hspace{0.2in} + W(\tilde{k} + 2,k) \theta(\tilde{k} + 1).
\end{align*}
Therefore, 
\begin{align}
    \|\psi_2(k)\|& \leq \theta(k) + \overline{w} \theta(k-1) + \dots + \overline{w} \theta(\tilde{k}+1) \nonumber \\
    & \leq \textstyle \overline{w}\sum_{s = \tilde{k}}^\infty \theta(s) \leq \frac{\hat{\varepsilon}}{2}\label{eq:inter_proof3}.
\end{align}

\noindent Using~(\ref{eq:equations_in_closed_form})-(\ref{eq:inter_proof3}) we have for all $k > \tilde{k}$,
\begin{align}
    \x(k+1) &\leq \underbrace{ \Big(1 + \sum_{s=0}^{\tilde{k}}  \theta(s)\Big) v(k) \mathbf{1}^\top }_{:= \chi_1(k)} \x(0) + \varepsilon \mathbf{1}\mathbf{1}^\top, \label{eq:x_up_bd}\\
    \x(k+1) &\geq \textstyle \Big(1 + \sum_{s=0}^{\tilde{k}}  \theta(s)\Big) v(k) \mathbf{1}^\top \x(0) - \varepsilon \mathbf{1}\mathbf{1}^\top. \label{eq:x_lw_bd}
\end{align}
Following the similar procedure for~(\ref{eq:proof_y_update}) we get, for $k \geq \tilde{k}$,
\begin{align}
    \y(k+1) &\leq \chi_1(k) \y(0)  + \varepsilon \mathbf{1}, \ \text{and}, \label{eq:y_up_bd} \\
    \y(k+1) &\geq \chi_1(k)  \y(0) - \varepsilon \mathbf{1}.  \label{eq:y_lw_bd}   
\end{align}
From~(\ref{eq:compact_ratio}), and~(\ref{eq:x_up_bd})-(\ref{eq:y_lw_bd}), for any $i \in [n]$,
\begin{align}\label{eq:z_bound}
  \textstyle \frac{[\chi_1(k) \x(0) - \varepsilon \mathbf{1}\mathbf{1}^\top ]_i}{[\chi_1(k) \y(0) + \varepsilon \mathbf{1}]_i} \leq 
   z_i(k+1) \leq \frac{[\chi_1(k) \x(0) + \varepsilon \mathbf{1}\mathbf{1}^\top]_i}{[\chi_1(k) \y(0) - \varepsilon \mathbf{1}]_i},
\end{align}
where, $[.]_{i}$ denote the $i^{th}$ row ($i^{th}$ entry) of the corresponding matrices (vectors). Denote $\overline{\theta} := 1 + \sum_{s=0}^{\tilde{k}} \theta(s)$, and $v_i(k)$ as the $i^{th}$ entry of $v(k)$. 

For $k > \tilde{k}$, for any $i \in [n]$,  from~(\ref{eq:z_bound})
\begin{align*}
    z_i(k+1) &\leq \frac{[\chi_1(k) \x(0) + \varepsilon \mathbf{1}\mathbf{1}^\top]_i}{[\chi_1(k) \y(0) - \varepsilon \mathbf{1}]_i} \\
    & = \frac{\overline{\theta} v_i(k) \mathbf{1}^\top \x(0) + \varepsilon \mathbf{1}^\top}{\overline{\theta} v_i(k) \mathbf{1}^\top \y(0) - \varepsilon }\\
    & = \frac{\frac{\mathbf{1}^\top \x(0)}{n} + \frac{\varepsilon}{nv_i(k) \overline{\theta}} \mathbf{1}^\top}{1 - \frac{\varepsilon}{nv_i(k) \overline{\theta}}} \\
    & = \frac{\mathbf{1}^\top \x(0)}{n} + \left(\frac{\mathbf{1}^\top \x(0)}{n} + \mathbf{1}^\top \right)\left(\frac{ \frac{\varepsilon}{nv_i(k) \overline{\theta}}}{1 - \frac{\varepsilon}{nv_i(k) \overline{\theta}}} \right) 
\end{align*}
Similarly, for $k > \tilde{k} $, for any $i \in [n]$, 
\begin{align*}
    z_i(k+1) &\geq \frac{[\chi_1(k) \x(0) - \varepsilon \mathbf{1}\mathbf{1}^\top]_i}{[\chi_1(k) \y(0) + \varepsilon \mathbf{1}]_i} \\
    & = \frac{\overline{\theta} v_i(k) \mathbf{1}^\top \x(0) - \varepsilon \mathbf{1}^\top}{\overline{\theta} v_i(k) \mathbf{1}^\top \y(0) + \varepsilon }\\
    & = \frac{\frac{\mathbf{1}^\top \x(0)}{n} - \frac{\varepsilon}{nv_i(k) \overline{\theta}} \mathbf{1}^\top}{1 + \frac{\varepsilon}{nv_i(k) \overline{\theta}}} \\
    & = \frac{\mathbf{1}^\top \x(0)}{n} - \left(\frac{\mathbf{1}^\top \x(0)}{n} + \mathbf{1}^\top \right)\left(\frac{ \frac{\varepsilon}{nv_i(k) \overline{\theta}}}{1 + \frac{\varepsilon}{nv_i(k) \overline{\theta}}} \right) 
\end{align*}
Let $\overline{\varepsilon}$ be such that, $\varepsilon \leq \frac{\overline{\varepsilon}}{(1 + \overline{\varepsilon})} (n\overline{\theta}v_i(k))$. Therefore, for all $k > \tilde{k}$, and for all $i \in [n]$,
\begin{align*}
   \textstyle  -  (\frac{\mathbf{1}^\top \x(0)}{n} + \mathbf{1}^\top)\overline{\varepsilon} \leq z_i(k) - \frac{\mathbf{1}^\top \x(0)}{n} \leq \textstyle (\frac{\mathbf{1}^\top \x(0)}{n} + \mathbf{1}^\top)\overline{\varepsilon},
\end{align*}
where, $\frac{\mathbf{1}^\top \x(0)}{n} = \frac{1}{n} \sum_{j=1}^n u_j$. Note that $\overline{\varepsilon}$ can be made arbitrarily small. Thus, we have the desired result.
\end{proof}

\begin{theorem}\label{thm:finite_time_estimate_mismatch}
Let Assumption~\ref{assp:strg_colmstoc} hold. Assume the edges in $\G(\V,\E)$ are \textbf{not} corrupted by additive noise in communication channels. Let $\{x_i(k), y_i(k), z_i(k)\}_{k \geq 0}^{i \in [n]}$ be the sequences generated by iterations~(\ref{eq:consensus_num})-(\ref{eq:consensus_ratio}). Let the sequence $\{\theta(k)\}_{k \geq 0}$ satisfy Assumption~\ref{assp:theta}. Then, for all $j \in [n]$ and $k \geq 0$, 
\begin{align*}
 \Big \| z_j(k) - \frac{1}{n} \sum_{i=1}^n u_i \Big \| \leq  \textstyle C_1 \left[ \lambda^k + \sum_{s=0}^{k-1} \theta(s) \lambda^{k-s-1} + \frac{1}{C} \theta(k) \right],
\end{align*}
where, $C, C_1$ and $\lambda \in (0,1)$ are some positive constants.
\end{theorem}
\begin{proof}
Recall~(\ref{eq:y_equation_in_closed_form}) and~(\ref{eq:equations_in_closed_form}), 
\begin{align*}
\y(k+1) = \textstyle W(0,k) \y(0) + \sum_{s=0}^{k} \theta(s) W(s+1,k)  \y(0),\\
\x(k+1) = \textstyle W(0,k) \x(0) + \sum_{s=0}^{k} \theta(s)  W(s+1,k) \x(0).
\end{align*}
Since, $W(l,k), k \geq l \geq 0$ is column stochastic we have,
\begin{align}
    \mathbf{1}^\top \x(k+1) &= \textstyle \mathbf{1}^\top W(0,k) \x(0) \nonumber \\ 
    & \hspace{0.3in} \textstyle + \sum_{s=0}^{k} \theta(s)  \mathbf{1}^\top W(s+1,k) \x(0) \nonumber \\
    & = \textstyle \mathbf{1}^\top \x(0) + \sum_{s=0}^{k} \theta(s) \mathbf{1}^\top \x(0)  . \label{eq:rate_closed_x_1}
\end{align}
From \cite{nedic2014distributed} Corollary 2, there exists stochastic vectors $v(k) \in \mathbb{R}^n$ and constants $C, \lambda \in (0,1)$ such that for any given $i$, for all $j = 1,\dots,n,$
\begin{align}\label{eq:corr_2}
    |W_{ij}(l,k) - v_i(k)| \leq C \lambda^{k -l}, \ \mbox{for all} \ k \geq l \geq 0.
\end{align}
Multiplying $\M(k+1)$ on both sides to~(\ref{eq:equations_in_closed_form}) and $v(k+1)$ to~(\ref{eq:rate_closed_x_1}) and subtracting we get, 
\begin{align}
    & (\M(k+1) - v(k+1) \mathbf{1}^\top) \x(k+1) \nonumber \\
    & =  (W(0,k+1) - v(k+1) \mathbf{1}^\top) \x(0) \nonumber \\
    & \hspace{0.5in} + \textstyle  \sum_{s=0}^k \theta(s)  (W(s+1,k+1) - v(k+1) \mathbf{1}^\top) \x(0) \nonumber \\
    & \implies \M(k+1) \x(k+1) = v(k+1) \mathbf{1}^\top \x(k+1) \nonumber \\
    & \hspace{0.5in} + (W(0,k+1) - v(k+1) \mathbf{1}^\top) \x(0) \nonumber \\
    & \hspace{0.1in} + \textstyle \sum_{s=0}^k (W(s+1,k+1) - v(k+1) \mathbf{1}^\top) \x(0) \theta(s). \label{eq:rate_closed_x}
\end{align}
Define $\mathbf{D}(l,k) := W(l,k) - v(k) \mathbf{1}^\top$. Replacing $k+1$ by $k$ in~(\ref{eq:rate_closed_x}) and adding $\theta(k) \x(0)$ on both sides. Comparing with~(\ref{eq:proof_x_update}) we get,
\begin{align}\label{eq:rate_closed_x2}
    \x(k+1) & = \M(k) \x(k) + \theta(k) \x(0) \nonumber \\
    & \hspace{-0.6in} = v(k) \mathbf{1}^\top \x(k) + \mathbf{D}(0,k) \x(0) \textstyle + \sum_{s=0}^{k-1} \mathbf{D}(s+1,k) \x(0) \theta(s) \nonumber \\
    & \hspace{-.4in} + \theta(k)\x(0).
\end{align}
Similarly, for~(\ref{eq:y_equation_in_closed_form})
\begin{align}\label{eq:rate_closed_y2}
    \y(k+1) & = \M(k) \y(k) + \theta(k) \y(0) \nonumber \\
    & \hspace{-0.6in} = v(k) \mathbf{1}^\top \y(k) + \mathbf{D}(0,k) \y(0) \textstyle + \sum_{s=0}^{k-1} \mathbf{D}(s+1,k) \y(0) \theta(s) \nonumber \\
    & \hspace{-.4in} + \theta(k)\y(0).
\end{align}
Using~(\ref{eq:compact_ratio}),~(\ref{eq:rate_closed_x2}), and~(\ref{eq:rate_closed_y2}) for any agent $i$ and $k \geq 0$, 
\begin{align*}
    & \textstyle  z_i(k+1) - \frac{\mathbf{1}^\top \x(0)}{n} \\
    & = \textstyle \frac{[v(k) \mathbf{1}^\top \x(k) + \mathbf{D}(0,k) \x(0)]_i + \sum_{s=0}^{k-1} [\mathbf{D}(s+1,k) \x(0) \theta(s)]_i + \theta(k) \x_i(0)}{y_i(k)}\\
    & \hspace{0.2in} - \textstyle \frac{\mathbf{1}^\top \x(0)}{n} \\
    & = \textstyle \frac{v_i(k) \mathbf{1}^\top \x(k) + \mathbf{D}_i(0,k) \x(0) + \sum_{s=0}^{k-1} \theta(s) \mathbf{D}_i(s+1,k) \x(0)+ \theta(k) \x_i(0)}{y_i(k)}\\
    & \hspace{0.2in} - \textstyle \frac{\mathbf{1}^\top \x(0)}{n} \\
    & \textstyle = \frac{n(v_i(k) \mathbf{1}^\top \x(k) + \mathbf{D}_i(0,k) \x(0) + \sum_{s=0}^{k-1} \theta(s) \mathbf{D}_i(s+1,k) \x(0)  + \theta(k) \x_i(0))}{n y_i(k)} \\
    & \hspace{0.1in} - \textstyle \frac{(v_i(k) \mathbf{1}^\top \y(k) + \mathbf{D}_i(0,k) \y(0) + \sum_{s=0}^{k-1} \theta(s) \mathbf{D}_i(s+1,k) \y(0))\mathbf{1}^\top \x(0) }{n y_i(k)} \\
    &  \hspace{0.1in} - \textstyle \frac{\theta(k) \y_i(0)\mathbf{1}^\top \x(0) }{n y_i(k)} \\
    & \textstyle = \frac{\mathbf{D}_i(0,k) (n \x(0) - \y(0) \mathbf{1}^\top \x(0)) }{n y_i(k)} + \frac{\theta(k)(n \x_i(0) - \y_i(0) \mathbf{1}^\top \x(0)) }{n y_i(k)} + \\
    & \hspace{0.2in}+ \textstyle \frac{\sum_{s=0}^{k-1} \theta(s)  \mathbf{D}_i(s+1,k) (n \x(0) - \y(0) \mathbf{1}^\top \x(0))}{n y_i(k)},
\end{align*}
where, $[.]_i$ denote the $i^{th}$ row of the corresponding matrix and in the last equality we used the fact $\frac{ nv_i(k) \mathbf{1}^\top \x(k)}{n y_i(k)} - \frac{v_i(k) \mathbf{1}^\top \y(k) \mathbf{1}^\top \x(0)}{n y_i(k)} = 0$.
Note $\| n \x(0) - \y(0) \mathbf{1}^\top \x(0) \| \leq \|n \I_n - \mathbf{1}\mathbf{1}^\top\| \|\x(0)\| \leq 2n^2 \|\x(0)\|$. From Lemma~\ref{lem:ratio_well_def}, there exists $y_o >0$ such that $\y_i(k) > y_o$ for all $i$ and $k \geq 0$. Therefore,
\begin{align*}
   & \textstyle \left \| z_i(k) - \frac{\mathbf{1}^\top \x(0)}{n} \right \| \\
   & \leq \textstyle  \frac{\|\mathbf{D}_i(0,k)\| \|n \x(0) - \y(0) \mathbf{1}^\top \x(0)\| }{n\|y_i(k)\|} + \frac{\theta(k)\|n \x_i(0) - \y_i(0) \mathbf{1}^\top \x(0)\| }{n\|y_i(k)\|} + \\
    & \hspace{0.2in}+ \textstyle \frac{\sum_{s=0}^{k-1} \theta(s) \| \mathbf{D}_i(s+1,k)\| \|n \x(0) - \y(0) \mathbf{1}^\top \x(0)\|}{n\|y_i(k)\|},\\
   & \textstyle \leq  \frac{2n^3 \|\x(0)\| C \lambda^{k}}{ny_o}  + \frac{2n^3 \|\x(0)\|\theta(k)}{ny_o}   + \frac{2n^3 \|\x(0)\| C \sum_{s=0}^{k-1} \theta(s) \lambda^{k-s-1}}{ny_o}\\
   & = \textstyle \frac{2n^2 \|\x(0)\| C}{y_o} \left[ \lambda^k + \sum_{s=0}^{k-1} \theta(s) \lambda^{k-s-1} + \frac{1}{C} \theta(k) \right].
\end{align*}
\end{proof}
\begin{corollary}
Let Assumption~\ref{assp:strg_colmstoc} hold. Assume the edges in $\G(\V,\E)$ are \textbf{not} corrupted by additive noise in communication channels. Let $\{x_i(k), y_i(k), z_i(k)\}_{k \geq 0}^{i \in [n]}$ be the sequences generated by the iterations~(\ref{eq:consensus_num})-(\ref{eq:consensus_ratio}). Let $\theta(k) = \rho^k$, for all $k \geq 1$ such that $\rho \in (0,1/\alpha)$, where $\alpha > 1$. Then, for all $j \in [n]$ and $k \geq 1$, 
\begin{align*}
    \textstyle \left \| z_j(k) - \frac{1}{n} \sum_{i=1}^n u_i \right \| \leq \Upsilon \varrho^k,
\end{align*}
where, $\Upsilon$ and $\varrho \in (0,1)$ are positive constants.
\end{corollary}
\begin{proof}
Let $ \varrho:= \max \{ \lambda, \alpha \rho \}$. 
From Theorem~\ref{thm:finite_time_estimate_mismatch}, 
\begin{align*}
    \big \| z_j(k) - \frac{1}{n} \sum_{i=1}^n u_i \big \| 
   & \textstyle \leq C_1 \left [ \lambda^k + \sum_{s=0}^{k-1} \theta(s) \lambda^{k-s-1}  + \frac{1}{C}\theta(k) \right] \\
   & \textstyle \hspace{-0.1in} \leq C_1  \left [ \lambda^k + \sum_{s=0}^{k-1} \lambda^{k-s-1} \rho^s + \frac{1}{C}\rho^k \right] \\
   & \textstyle \hspace{-0.1in}  \leq C_1 \left [ \lambda^k + \sum_{s=0}^{k-1} \varrho^{k-s-1} \rho^s + \frac{1}{C}\rho^k \right] \\
   & \textstyle \hspace{-0.1in} \leq C_1  \left [ \varrho^k + \frac{\varrho^k}{\varrho} \sum_{s=0}^{k-1} (\frac{\rho}{\varrho})^s + \frac{1}{C}\varrho^k \right] \\
   & \textstyle \hspace{-0.1in} \leq C_1 \left [ 1 + \frac{1}{\varrho}\sum_{s=0}^{\infty} (\frac{\rho}{\varrho})^s + \frac{1}{C} \right]  \varrho^k \\
   & \textstyle \hspace{-0.1in} = C_1  \left[ 1 + \frac{1}{\varrho - \rho} + \frac{1}{C} \right]  \varrho^k = \Upsilon \varrho^k,
\end{align*}
where, $\Upsilon :=  C_1  \left[ 1 + \frac{1}{\varrho - \rho} + \frac{1}{C} \right]$.
\end{proof}

\begin{subsection}{Convergence under noisy communication channels}
Here, we analyze the performance of the proposed $\nrps$ algorithm under the communication noise. We first introduce conditions on the additive noise in the communication links. In particular, we make the following assumption:
\begin{assump}\label{assp:noise}
The noise in communication channels in $\G(\V,\E)$ is bounded, i.e. there exists $\delta_{i,j}$ such that $|\eta_{x_{ij}}(k)| \leq \delta_{i,j}, |\eta_{y_{ij}}(k)| \leq \delta_{i,j}$ for $i,j \in [n]$ and $k \geq 1$.
\end{assump}

\begin{lemma}\label{lem: y_well_def_in_noise}
Let Assumption~\ref{assp:strg_colmstoc} hold. Let sequence $\{\beta(k)\}_{k \geq 0}$ satisfy Assumption~\ref{assp:beta}. Assume that the channel noise satisfy Assumption~\ref{assp:noise}. Then there exist sequences $\{\theta(k)\}_{k \geq 0}$ satisfying Assumption~\ref{assp:theta} and $\tilde{y}_o > 0$ such that $y_i(k) > \tilde{y}_o$ for all $k \geq 1$ and $i \in [n]$, i.e., the estimate $y_i(k)$ remain positive at all time-instants and the states $z_i(k)$ are well-defined for all $k, i \in [n]$.
\end{lemma}
\begin{proof}
Let $\delta := \max\limits_{i,j} \delta_{ij}$. Consider a family $\Theta$ of sequences with any element $\theta(k) \in \Theta$ such that,
\begin{align}\label{eq:theta_family}
    \theta(k) = \begin{cases}
                 d_1 n \delta \beta(k) \in \mathbb{R}, & \ k \leq K_\theta, \\
                  \textstyle \frac{d_2nb\delta}{k^q} \in [0,1],  & \ k > K_\theta,
               \end{cases}
\end{align}
where, $ 1 \leq d_1 < \infty$ and $ 1 \leq d_2 \leq K_\theta^q / nb\delta$ and $K_{\theta} \geq \max \{K_\beta, (n\delta b)^{1/q} \}$, where $K_\beta$ is the time index in Assumption~\ref{assp:beta}, such that $a k^{-q} \leq \beta(k) \leq b k^{-q}$ for all $k \geq K_\beta$, with $q > 1$ and $0 \leq a \leq b < \infty$. Different choices of $d_1$ and $d_2$ lead to different sequences in $\Theta$. Note that every sequence $\theta(k)$ in the family $\Theta$ satisfy Assumption~\ref{assp:theta}. Using~(\ref{eq:compact_den}) and proceeding as in~(\ref{eq:y_equation_in_closed_form}) we have,
\begin{align}\label{eq:y_for_proof_with_noise}
     \y(k+1) & = \textstyle W(0,k) \y(0) + \sum_{s=0}^k  \theta(s) W(s+1,k) \y(0) \nonumber \\
    & \hspace{0.2in} \textstyle  + \sum_{s=0}^k \beta(s) W(s+1,k)  \eta_y(s) \\
    & \hspace{-0.55in}  = \textstyle W(0,k) \y(0) + \sum_{s=0}^k W(s+1,k) (\theta(s)  \y(0) + \beta(s) \eta_y(s)). \nonumber
\end{align}

\textit{Claim:} Let $\theta(k) \in \Theta$, then $(\theta(k)  \y(0) + \beta(k) \eta_y(k)) > 0$ for all $k \geq 0$. \\
\textit{Proof:} Let $k \leq K_\theta$, 
\begin{align*}
    \theta(k)  \y(0) + \beta(k) \eta_y(k) &= d_1 n \delta \beta(k) \y(0) + \beta(k) \eta_y(k) \\
    & \hspace{-1in} \geq d_1 n \delta \beta(k) \mathbf{1}^\top - \delta \beta(k)  \mathbf{1}^\top = ( d_1 n - 1)\delta \beta(k) \mathbf{1}^\top > 0.
\end{align*}
For $k > K_\theta$, 
\begin{align*}
    \theta(k)  \y(0) + \beta(k) \eta_y(k) &= \textstyle \frac{d_2 n b \delta}{k^q} \y(0) + \beta(k) \eta_y(k) \\
    & \hspace{-0.4in} \geq \textstyle \frac{d_2 n b \delta}{k^q} \mathbf{1}^\top - \frac{b \delta}{k^q}  \mathbf{1}^\top = \textstyle ( d_2 n - 1) \frac{b \delta}{k^q} \mathbf{1}^\top > 0. \hspace{0.4in} \qed
\end{align*} 

Let the sequence $\theta(k) \in \Theta$ as defined in~(\ref{eq:theta_family}). We follow the same steps in the proof of Lemma~\ref{lem:ratio_well_def}. Given, $\varepsilon >0$, for any  $i \in [n]$ and $k > \ell_\varepsilon$,
\begin{align*}
 y_i(k+1) &= \textstyle \sum_{j=1}^n W_{ij}(0,k) y_j(0) \\
 & \hspace{-0.2in} \textstyle + \sum_{j=1}^n \sum_{s=0}^k W_{ij}(s+1,k) (\theta(s)  y_j(0) + \beta(s) \eta_{y_j}(s)) \\
 & > \textstyle  \sum_{j=1}^n W_{ij}(0,k) \\
 & = \textstyle  \sum_{j=1}^n \big[v(\ell_\varepsilon)\mathbf{1}^\top + e(\ell_\varepsilon) W(\ell_\varepsilon + 1, k) \big]_{ij} \\
 & > \textstyle  \sum_{j=1}^n \big[v(\ell_\varepsilon)\mathbf{1}^\top - \frac{\varepsilon}{2} \mathbf{1}\mathbf{1}^\top W(\ell_\varepsilon + 1, k) \big]_{ij} \\
 & > \textstyle  \sum_{j=1}^n \big[v(\ell_\varepsilon)\mathbf{1}^\top - \frac{\varepsilon}{2} \mathbf{1}\mathbf{1}^\top \big]_{ij} \\
 & = \textstyle  n v_i(\ell_\varepsilon) - n\frac{\varepsilon}{2} \\
 & >   n \min_{1 \leq i \leq n} v_i(\ell_\varepsilon) - \textstyle \varepsilon\frac{n}{2} > 0,
\end{align*}
where, we used~(\ref{eq:y_sum_bound}), the claim above and the fact that $W(l,k)$ is column stochastic for all $k,l \geq 0$.

For $k \leq \ell_\varepsilon$, from~(\ref{eq:y_sum_bound}) we have
\begin{align*}
 y_i(k) & > \textstyle w_o^{k-1} > \theta(0)w_o^{\ell_\varepsilon} > 0.
\end{align*}
Therefore, $y_i(k) > \tilde{y}_o : =\min \{n \min\limits_{1 \leq i \leq n} v_i(\ell_\varepsilon), \theta(0) w_o^{\ell_\varepsilon}\} \ \forall k \geq 1$ with any $\theta(k) \in \Theta$ and $\beta(k)$ satisfying Assumption~\ref{assp:beta}. 
\end{proof}

\begin{theorem}\label{thm:withnoise}
Let Assumption~\ref{assp:strg_colmstoc} hold. Let sequence $\{\beta(k)\}_{k \geq 0}$ satisfy Assumption~\ref{assp:beta}. Assume that the channel noise satisfy Assumption~\ref{assp:noise}. Let $\{z_i(k)\}_{k \geq 0}^{i \in [n]}$ be the sequences generated by iteration~(\ref{eq:consensus_ratio}). Given $\mu > 0$, let $\{\theta(k)\}_{k \geq 0}$ be a sequence from the family $\Theta$ defined in Lemma~\ref{lem: y_well_def_in_noise} with $\theta(0) := n\delta \sum_{s=0}^\infty \beta(s) / \mu$, where, $\delta = \max\limits_{i,j} \delta_{ij}$. Then, for all $i \in [n]$, \textit{almost surely},
\begin{align*}
\textstyle \frac{\frac{1}{n} \sum_{j=1}^n u_j - \mu}{1 + \mu} \leq \lim_{k \to \infty} z_i(k) \leq \frac{\frac{1}{n} \sum_{j=1}^n u_j + \mu}{1 - \mu}.
\end{align*}
\end{theorem}

\begin{proof}

Recall~(\ref{eq:compact_num}) and proceed as in~(\ref{eq:equations_in_closed_form}),
\begin{align}\label{eq:x_eq_with_noise}
     \x(k+1) & = \underbrace{W(0,k) \x(0)}_{\text{first}} + \underbrace{\sum_{s=0}^k  \theta(s) W(s+1,k) \x(0)}_{\text{second}}  \nonumber \\
    & \hspace{0.2in} + \underbrace{\sum_{s=0}^k \beta(s) W(s+1,k)  \eta_x(s)}_{\text{third}} 
\end{align}

Note that as $W(l,k)$ is column stochastic for all $k,l \geq 1$, $\|W(l,k)\| \leq \overline{w} < \infty$ for all $k,l$. 

Under Assumption~\ref{assp:theta} and from Theorem~\ref{thm:wolfowitz}, given $\varepsilon > 0$. Define $\hat{\varepsilon} := \frac{\varepsilon}{\max \{ \|\x(0)\|,n \}}$. There exists $\tilde{k}$ and $v(k) \in \mathbb{R}^n$ such that $\sum_{s = \tilde{k}}^\infty \theta(s) < \frac{\hat{\varepsilon}}{3\overline{w}}, \sum_{s = \tilde{k}}^\infty \beta(s) \leq \frac{\hat{\varepsilon}}{3 n \delta \overline{w}} $ and $W(l,k) = v(k)\mathbf{1}^\top + e(k)$, for all $k > \tilde{k} + l$, where, $|e_{ij}(k)| \leq \frac{\hat{\varepsilon}}{3}$ for all $i,j \in [n]$.

Thus the first and second term in~(\ref{eq:x_eq_with_noise}) can be re-formulated in the same manner as~(\ref{eq:x1})-(\ref{eq:inter_proof3}). Consider, the third term in~(\ref{eq:x_eq_with_noise}),  
\begin{align}\label{eq:noise_inter_proof1}
    \sum_{s=0}^{k} \beta(s) W(s+1,k)  \eta_x(s) & = \underbrace{\sum_{s=0}^{\tilde{k}} \beta(s) W(s+1,k)  \eta_x(s)}_{:= \zeta_1(k)} \nonumber \\ 
    & \hspace{-0.4in} + \underbrace{\sum_{s= \tilde{k} + 1 }^{k}\beta(s) W(s+1,k)  \eta_x(s)}_{:= \zeta_{x}(k)}
\end{align}
Note, $k > \tilde{k} + s + 1$ for all the terms in $\zeta_1(k)$. Substituting $W(s+1,k) := v(k)\mathbf{1}^\top + \tilde{e}(k)$, with $|\tilde{e}_{ij}(k)| < \hat{\varepsilon}/3$ for all $i,j \in [n]$ we get,
\begin{align}\label{eq:noise_inter_proof2}
    \textstyle \zeta_1(k) = \sum_{s=0}^{\tilde{k}} \beta(s) [v(k)\mathbf{1}^\top + \tilde{e}(k)]\eta_x(s).
\end{align}
Now consider $\zeta_x(k)$, 
\begin{align*}
    \zeta_x(k) & = \textstyle \sum_{s= \tilde{k} + 1 }^{k}\beta(s) W(s+1,k)  \eta_x(s) \\ 
    & = \beta(k) \eta_x(k) +  \beta(k-1) \M(k) \eta_x(k-1) + \dots  \\ 
    & \hspace{0.2in} + \beta(\tilde{k} +1) W(\tilde{k} + 2,k) \eta_x(\tilde{k} +1).
\end{align*}
Therefore, 
\begin{align}
    \|\zeta_x(k)\|& \leq n \delta  \beta(k) +  n \delta \overline{w} \beta(k-1) \dots + n \delta \overline{w} \beta(\tilde{k} +1) \nonumber \\
    & \leq \textstyle n \delta \overline{w} \sum_{s = \tilde{k}}^\infty \beta(s) \leq \frac{\hat{\varepsilon}}{3}. \label{eq:noise_inter_proof3}
\end{align}
Using~(\ref{eq:x_eq_with_noise}), (\ref{eq:x1})-(\ref{eq:inter_proof3}), and~(\ref{eq:noise_inter_proof1})-(\ref{eq:noise_inter_proof3}), for all $k > \tilde{k}$,
\begin{align}
    \x(k+1) &\leq v(k) \mathbf{1}^\top (\underbrace{ \x(0)  + \sum_{s=0}^{\tilde{k}}  \theta(s) \x(0)  + N_x(\tilde{k}) }_{:= \chi_x(k)}) + \varepsilon \mathbf{1}\mathbf{1}^\top,  \label{eq:noise_x_up_bd}\\
    \x(k+1) &\geq v(k) \mathbf{1}^\top (\x(0)  + \sum_{s=0}^{\tilde{k}}  \theta(s) \x(0)  + N_x(\tilde{k})) - \varepsilon \mathbf{1}\mathbf{1}^\top \label{eq:noise_x_lw_bd},
\end{align}
where, $N_x(\tilde{k}) := \sum_{s=0}^{\tilde{k}} \beta(s)\eta_x(s)$,$N_y(\tilde{k}) := \sum_{s=0}^{\tilde{k}} \beta(s)\eta_y(s)$. Following the similar procedure for~(\ref{eq:y_for_proof_with_noise}) we get, for $k \geq 1$,
\begin{align}
    & \y(k+1) \leq v(k) \mathbf{1}^\top \chi_y(k)  + \varepsilon \mathbf{1}, 
    \label{eq:noise_y_up_bd} \\
    & \y(k+1) \geq v(k) \mathbf{1}^\top \chi_y(k) - \varepsilon \mathbf{1}.  \label{eq:noise_y_lw_bd}   
\end{align}
From~(\ref{eq:compact_ratio}), and~(\ref{eq:noise_x_up_bd})-(\ref{eq:noise_y_lw_bd}), 
\begin{align}\label{eq:noise_z_bound}
  \hspace{-0.1in} \textstyle \frac{[v(k) \mathbf{1}^\top \chi_x(k)  - \varepsilon \mathbf{1}\mathbf{1}^\top]_i}{[v(k) \mathbf{1}^\top \chi_y(k)  + \varepsilon \mathbf{1}]_i} \leq 
   z_i(k+1) \leq \frac{[v(k) \mathbf{1}^\top \chi_x(k)  + \varepsilon \mathbf{1}\mathbf{1}^\top]_i}{[v(k) \mathbf{1}^\top \chi_y(k)  -\varepsilon \mathbf{1}]_i},
\end{align}
where, $[.]_{i}$ denote the $i^{th}$ row ($i^{th}$ entry) of the corresponding matrices (vectors). Denote $\overline{\theta} := 1 + \sum_{s=0}^{\tilde{k}} \theta(s)$, and $v_i(k)$ as the $i^{th}$ entry of $v(k)$.

For $k > \tilde{k}$, for any $i \in [n]$, from~(\ref{eq:noise_z_bound}), 
\begin{align*}
    z_i(k+1) &\leq  \frac{[v(k) \mathbf{1}^\top \chi_x(k)  + \varepsilon \mathbf{1}\mathbf{1}^\top]_i}{[v(k) \mathbf{1}^\top \chi_y(k)  -\varepsilon \mathbf{1}]_i}  \\
    & \hspace{ -0.4in} \leq \frac{v_i(k) \mathbf{1}^\top \x(0) + \overline{\theta} v_i(k) \mathbf{1}^\top \x(0) + v_i(k) \mathbf{1}^\top N_x(\tilde{k}) + \varepsilon \mathbf{1}^\top }{v_i(k) \mathbf{1}^\top \y(0) + \overline{\theta} v_i(k) \mathbf{1}^\top \y(0) + v_i(k) \mathbf{1}^\top N_y(\tilde{k}) - \varepsilon }\\
    & \leq \frac{\frac{\mathbf{1}^\top \x(0)}{n} + \frac{\overline{\theta} \mathbf{1}^\top \x(0)}{n} +  \frac{\mathbf{1}^\top N_x(\tilde{k})}{n} + \varepsilon \mathbf{1}^\top }{ 1 + \overline{\theta} + \frac{\mathbf{1}^\top N_y(\tilde{k})}{n} - \varepsilon} \\
    & \leq \frac{\frac{\mathbf{1}^\top \x(0)}{n} \left( 1 + \overline{\theta} \right) +  \frac{\mathbf{1}^\top N_x(\tilde{k})}{n} + \varepsilon \mathbf{1}^\top }{ 1 + \overline{\theta} + \frac{\mathbf{1}^\top N_y(\tilde{k})}{n} - \varepsilon}\\
    & \leq \frac{\frac{\mathbf{1}^\top \x(0)}{n} +  \frac{\mathbf{1}^\top N_x(\tilde{k})}{n\left( 1 + \overline{\theta} \right)} + \frac{\varepsilon}{\left( 1 + \overline{\theta} \right)} \mathbf{1}^\top }{ 1 + \frac{\mathbf{1}^\top N_y(\tilde{k})}{n \left( 1 + \overline{\theta} \right)} - \frac{\varepsilon}{\left( 1 + \overline{\theta} \right)}} \\
    & \leq \frac{a}{b} + \left( \frac{a}{b} + \mathbf{1}^\top \right) \left(\frac{\frac{\varepsilon}{b (1+ \overline{\theta})}}{1 + \frac{\varepsilon}{b (1+ \overline{\theta})}}\right).
\end{align*}
where, $a := \frac{\mathbf{1}^\top \x(0)}{n} +  \frac{\mathbf{1}^\top N_x(\tilde{k})}{n\left( 1 + \overline{\theta} \right)}, b := 1 + \frac{\mathbf{1}^\top N_y(\tilde{k})}{n \left( 1 + \overline{\theta} \right)}$
Similarly, for $k > \tilde{k}$, for any $i \in [n]$,
\begin{align*}
    z_i(k+1) &\geq \frac{a}{b} - \left( \frac{a}{b} + \mathbf{1}^\top \right) \left(\frac{\frac{\varepsilon}{b (1+ \overline{\theta})}}{1 - \frac{\varepsilon}{b (1+ \overline{\theta})}}\right)
\end{align*}
Let $\overline{\varepsilon}$ be such that, $\varepsilon \leq \frac{\overline{\varepsilon}}{b(1 + \overline{\varepsilon})} (1+ \overline{\theta})$. Therefore, for all $k > \tilde{k}$, and for all $i \in [n]$,
\begin{align*}
   & - \left( \frac{a}{b} + \mathbf{1}^\top \right) \overline{\varepsilon} \leq z_i(k) - \frac{\frac{\mathbf{1}^\top \x(0)}{n} +  \frac{\mathbf{1}^\top N_x(\tilde{k})}{n\left( 1 + \overline{\theta} \right)}}{1 + \frac{\mathbf{1}^\top N_y(\tilde{k})}{n \left( 1 + \overline{\theta} \right)}}  \leq \left( \frac{a}{b} + \mathbf{1}^\top \right) \overline{\varepsilon} .
\end{align*}
Note that $\overline{\varepsilon}$ can be made arbitrarily small. Thus, for all $i \in [n]$,
\begin{align*}
   \textstyle \lim_{k \to \infty} z_i(k) = \frac{\frac{\mathbf{1}^\top \x(0)}{n} +  \frac{\mathbf{1}^\top N_x(\tilde{k})}{n\left( 1 + \overline{\theta} \right)}}{1 + \frac{\mathbf{1}^\top N_y(\tilde{k})}{n \left( 1 + \overline{\theta} \right)}} := z_\infty.
\end{align*}
Moreover, under assumptions~\ref{assp:beta}-\ref{assp:noise},
\begin{align*}
   \textstyle \frac{ \frac{\sum_{j=1}^n u_j}{n} - \frac{n \delta \sum_{s=0}^\infty \beta(s)}{\left(1 + \sum_{s=0}^{\infty} \theta(s) \right)} \mathbf{1}^\top }{ 1 + \frac{\delta \sum_{s=0}^\infty \beta(s)}{\left(1 + \sum_{s=0}^{\infty} \theta(s) \right)} } \leq z_{\infty} \leq \frac{ \frac{\sum_{j=1}^n u_j}{n} + \frac{n \delta \sum_{s=0}^\infty \beta(s)}{\left(1 + \sum_{s=0}^{\infty} \theta(s) \right)} \mathbf{1}^\top}{ 1 - \frac{\delta \sum_{s=0}^\infty \beta(s)}{\left(1 + \sum_{s=0}^{\infty} \theta(s) \right)} },
\end{align*}
where, we used $\frac{\mathbf{1}^\top \x(0)}{n} = \frac{1}{n} \sum_{j=1}^n u_j$. Note, that if given $\mu > 0$, if $\theta(0) := n \delta \sum_{s=0}^\infty \beta(s)/ \mu$, we have
\begin{align*}
   \textstyle \frac{ \frac{\sum_{j=1}^n u_j}{n} - \mu }{ 1 + \mu} \leq z_{\infty} \leq \frac{ \frac{\sum_{j=1}^n u_j}{n} + \mu}{ 1 - \mu}.
\end{align*}
\end{proof}

\begin{remark}
Theorem~\ref{thm:withnoise} states that the estimates $z_i(k)$ converge to the same value $z_\infty$ (achieve consensus) under the additive communication noise. The upper and lower bounds on $z_\infty$ can be made very close to the average $\overline{u}$ by appropriately choosing the sequences $\theta(k)$ and $\beta(k)$. In particular, if $\sum_{s=0}^\infty \theta(s) >> \sum_{s=0}^\infty \beta(s)$ then $z_\infty$ is very close to the actual average $\overline{u}$. This can be achieved via suitable choices of $K_\theta$ and the starting values of the sequence $\theta(k), k < K_\theta$ in~(\ref{eq:theta_family}). We demonstrate this using numerical simulation studies with different noise realizations in Section~\ref{sec:sim_results}.
\end{remark}

\end{subsection}

\end{section}

\begin{section}{Simulation Study}\label{sec:sim_results} 
In this section, we demonstrate the performance of the proposed $\nrps$ algorithm. A network of 10 agents generated using the Erdos-Renyi model \cite{erdHos1960evolution} with a connectivity probability of 0.35 is considered for demonstration. The initial values of the agents chosen for the numerical study are $u = [1,2,3,4,5,6,7,8,9,10]$ with the average value being $ \frac{1}{n} \sum_{j=1}^n u_j = 5.5$. The examples in the numerical study are implemented in MATLAB, and run on a desktop computer with 16 GB RAM and an Intel Core i7 processor running at 1.90 GHz. The noise realizations are obtained by simulating random variables that are drawn independently from a probability distribution at each iteration of the algorithm under study. We present the plots of the evolution of agents' estimates against the algorithm iterations. 

First, we present the performance of the proposed $\nrps$ algorithm under perfect communication. Then we present the performance under imperfect communication via the following case studies: (i) performance of $\nrps$ algorithm under different algorithm parameters, (ii) performance under different magnitudes of the additive noise and (iii) robustness of the $\nrps$ algorithm under the violation of the upper bound in Assumption~\ref{assp:noise} during the run-time of the algorithm. Lastly, we present a comparison between the proposed $\nrps$ algorithm and other existing algorithms in the literature.

\subsection{Performance under perfect communication}
Fig.~\ref{fig:without_noise_den} presents the agent state trajectories for the $\nrps$ algorithm when the communication links are noiseless. The algorithm parameters are: $\theta(k) = 0.7^k$ for all $k \geq 1$ and $K_\beta = 200$ such that  $\beta(k) = 0.35, \forall k < 200$ and $\beta(k) = 100/k^{1.1}, \forall k \geq 200$. Figs.~\ref{fig:without_noise_den}(a) and~\ref{fig:without_noise_den}(b) show the estimates $x_i(k)$ and $y_i(k)$ (called the numerator and denominator states respectively) for all the agents. The numerator and denominator states reach a steady state. The agent states $z_i(k)$ are plotted in Fig.~\ref{fig:without_noise_den}(c); it can be seen that the agents states $z_i(k)$ achieve consensus to the average value of $5.5$.
\begin{figure}[h] 
\centering
    \includegraphics[scale=0.2,trim={0.35cm 5.2cm 1.1cm 6.1cm},clip] {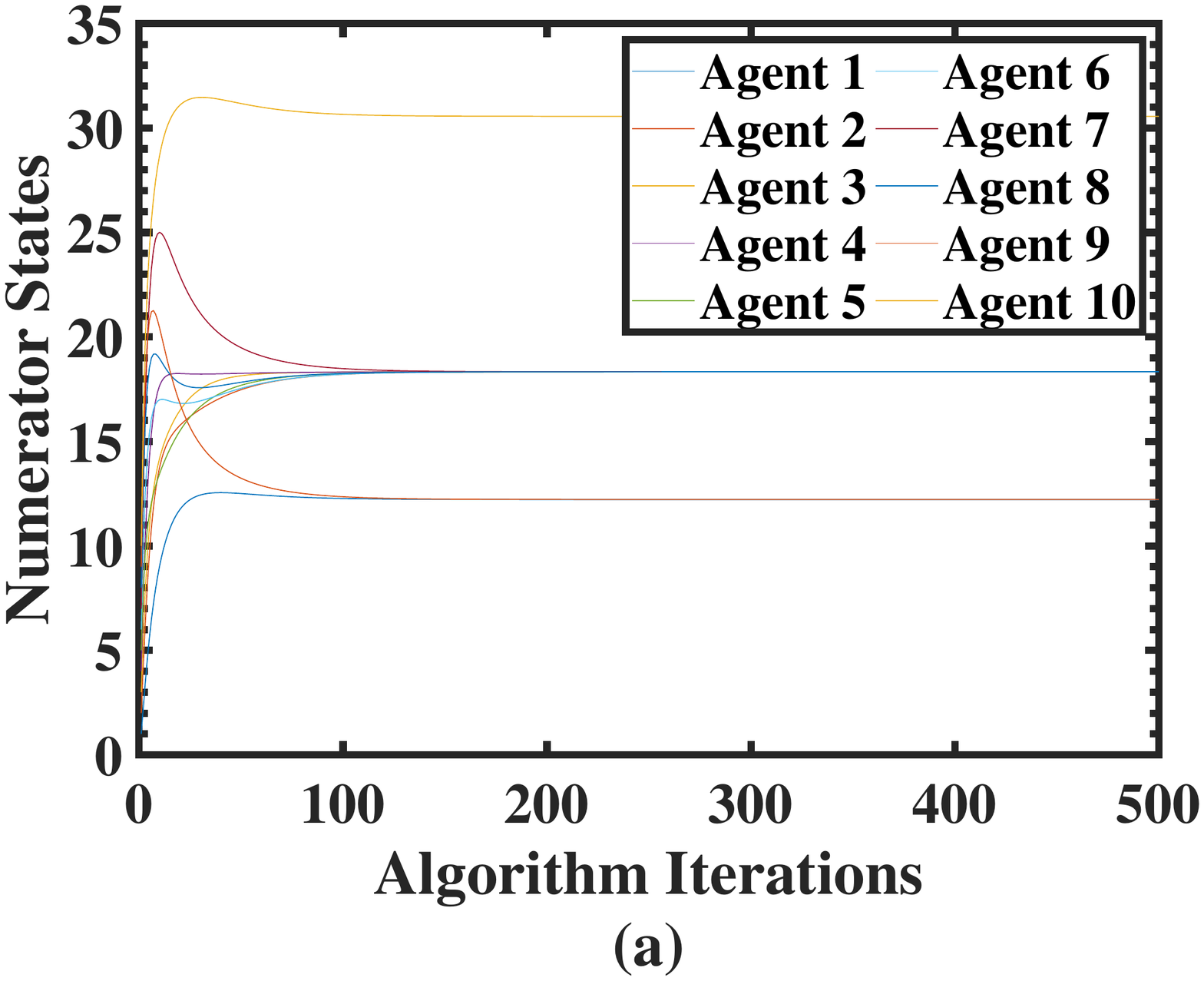}
    \includegraphics[scale=0.2,trim={0.35cm 5.2cm 1.1cm 6.1cm},clip] {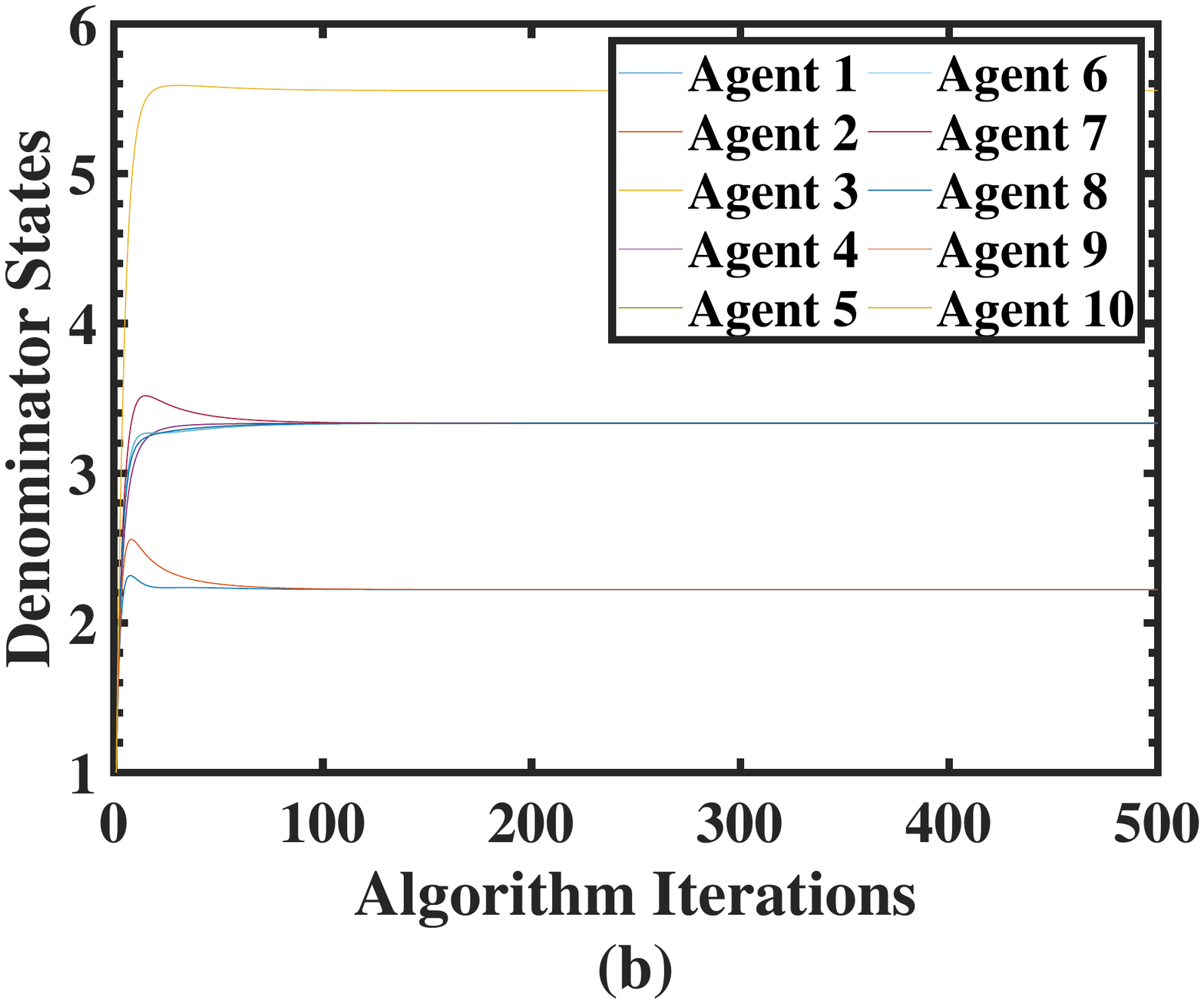}\\
    \includegraphics[scale=0.2,trim={0.35cm 5.2cm 1.1cm 6.1cm},clip] {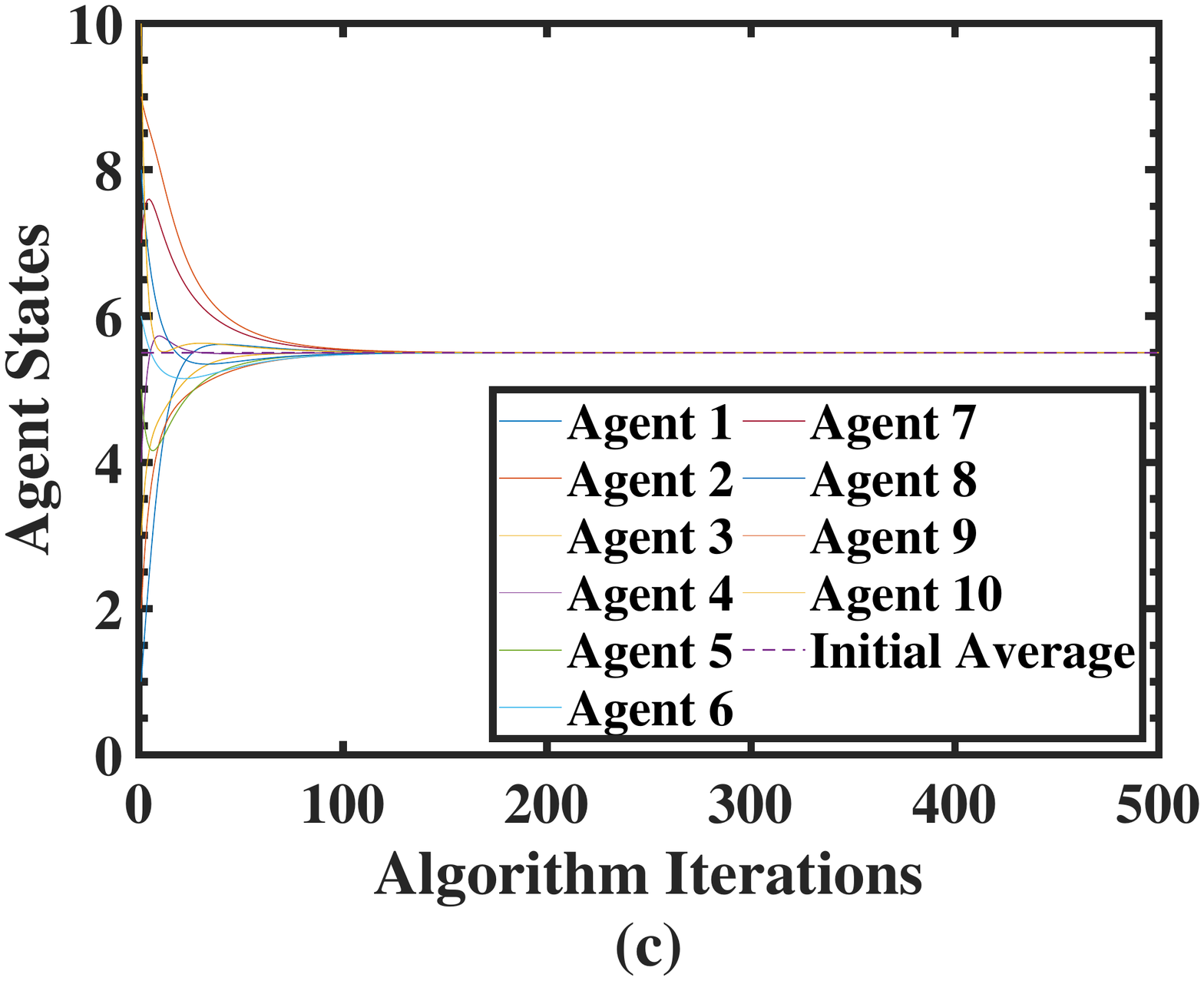} 
  \caption{The agent state achieve consensus to the average value with the communication links being noiseless.}
  \label{fig:without_noise_den} 
\end{figure}

\begin{figure}[b] 
\centering
    \includegraphics[scale=0.195,trim={5.8cm 0.2cm 5.9cm 0.35cm},clip] {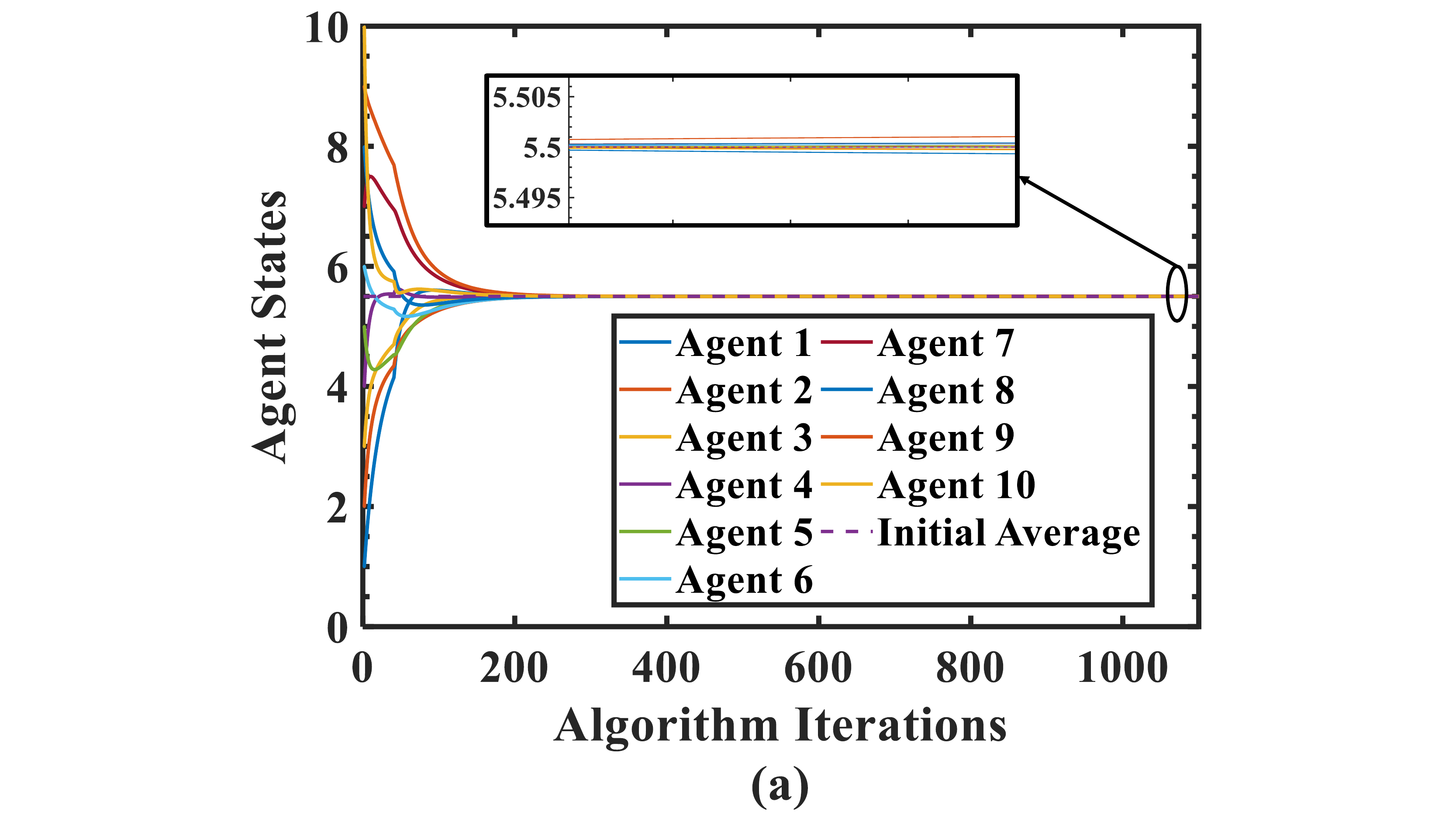}
    \includegraphics[scale=0.195,trim={6.1cm 0.2cm 5.8cm 0.45cm},clip] {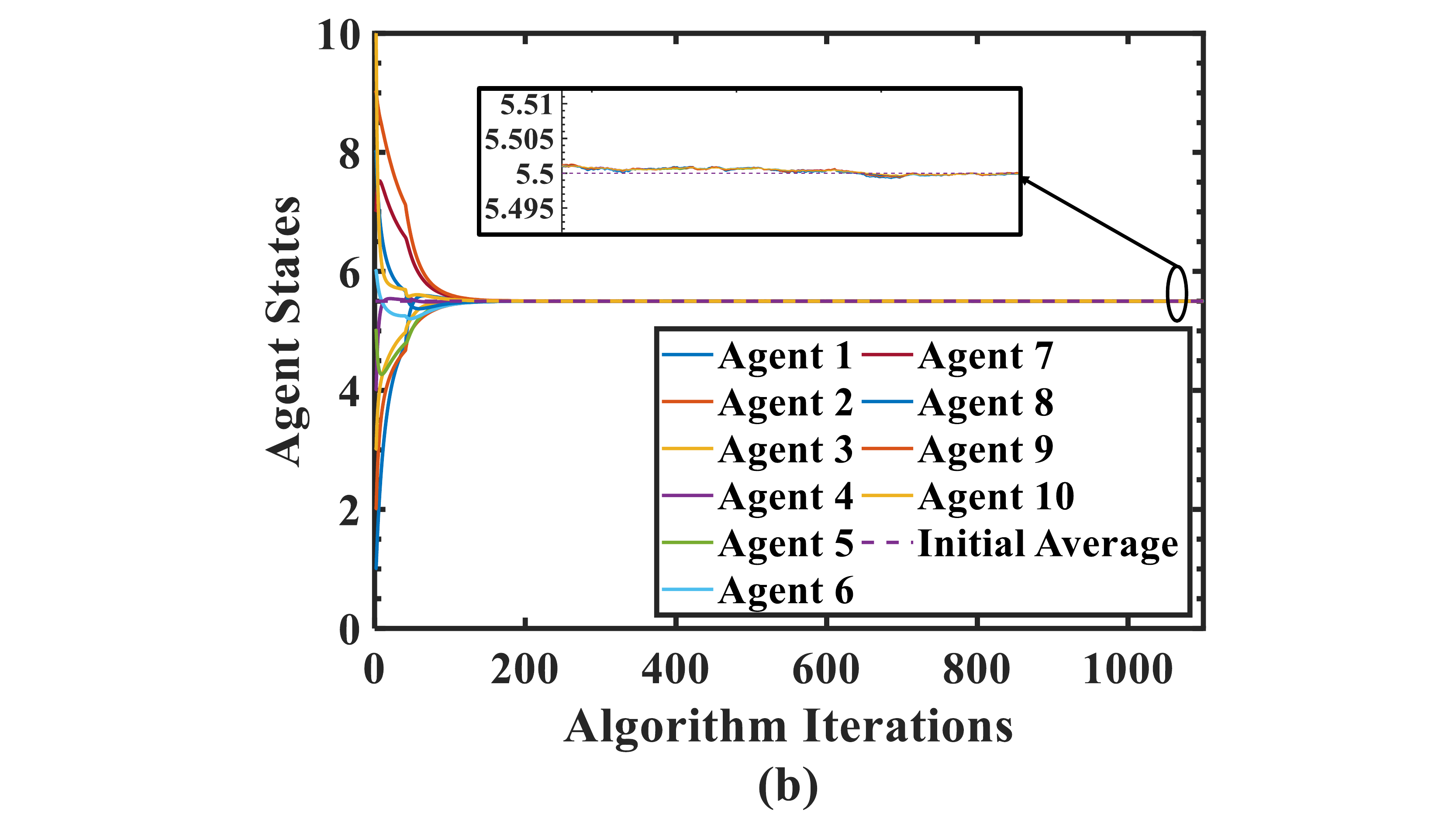}
  \caption{Performance of the $\nrps$ algorithm under noisy communication with different choices of algorithm parameters: (a) Row~(a) of Table~\ref{tab:parameters}, (b) Row~(b) of Table~\ref{tab:parameters}.}
  \label{fig:beta_theta_seq} 
\end{figure}

\subsection{Performance under communication noise}
In this section we present the performance of the proposed $\nrps$ algorithm under noisy communication links. \\
\noindent \textbf{1. Different choices of algorithm parameters:} Here, we demonstrate that the $\nrps$ algorithm achieves average consensus under the additive noise. The noise in communication links is realized from a uniform distribution  $\mathcal{U}(-1,1)$. We employ two different choices of $\beta(k)$ sequences and accordingly choose the $\theta(k)$ sequences from the family of sequences in~(\ref{eq:theta_family}). The choices are summarized in Table~\ref{tab:parameters}.
\begin{table}[h]
\centering
\captionsize
\caption{Example sequences for $\nrps$}
\begin{tabular}{c|c|c} 
\hline 
(a) & $\beta(k) = 0.2, \forall k < 500,$ & $\theta(k) = 100, \forall k < 500,$ \\
& $\beta(k) = 1/k^{1.5}, \forall k \geq 500 $  & $\theta(k) = 10/k^{1.5},  \forall k \geq 500$ \\[0.5ex] \hline
(b) & $\beta(k) = 0.35, \forall k < 200,$  & $\theta(k) = 100, \forall k < 200,$  \\
& $ \beta(k) = 100/k^{1.1}, \forall k \geq 200 $ & $\theta(k) = 150/k^{1.1},  \forall k \geq 200 $ \\[0.5ex]  
\end{tabular}
\label{tab:parameters}
\end{table}

Figs.~\ref{fig:beta_theta_seq}(a) and~\ref{fig:beta_theta_seq}(b) present the agent state trajectories for the $\nrps$ algorithm realizations with sequences in Table~\ref{tab:parameters} rows (a) and~(b) respectively under the additive noise. The inset in the Figs.~\ref{fig:beta_theta_seq}(a) and~\ref{fig:beta_theta_seq}(b) show that the proposed $\nrps$ algorithm has strong performance leading to consensus among the agents under additive noise. Moreover, the final consensus value of the agents is close to the initial average of $5.5$. \\

\noindent \textbf{2. Different magnitudes of the additive noise:} 
The effect of the magnitude of the communication noise relative to the initial values of the agents is presented here. The noise is realized as a uniform random variable. For the numerical experiment the magnitude of the noise realizations are scaled to: (i) $5$ times the original value, i.e. $\mathcal{U}(-5,5)$ and (ii) $10$ times the original value i.e. $\mathcal{U}(-10,10)$. We choose the sequence $\beta(k) = 0.35, \forall k < 200, \beta(k) = 1/k^{1.2}, \forall k \geq 200$. The sequence $\theta(k)$ is modified according to the noise upper bound. For the two cases, the following $\theta(k)$ sequences are chosen in accordance to the family~(\ref{eq:theta_family}) for the $\nrps$  algorithm: (i) $\theta(k) = 100, \forall k < 200, \theta(k) = 50/k^{1.2}, \forall k \geq 200$ and (ii) $\theta(k) = 100, \forall k < 200, \theta(k) = 100/k^{1.2}, \forall k \geq 200$. Figs.~\ref{fig:noise_mag}(a) and~\ref{fig:noise_mag}(b) provide the state trajectories of different agents under the $\nrps$ algorithm for cases (i) and (ii) respectively. It can be seen that the agent states achieve consensus close to the initial average value of $5.5$.
\begin{figure}[h] 
\centering
    \includegraphics[scale=0.195,trim={5.8cm 0.2cm 5.9cm 0.45cm},clip] {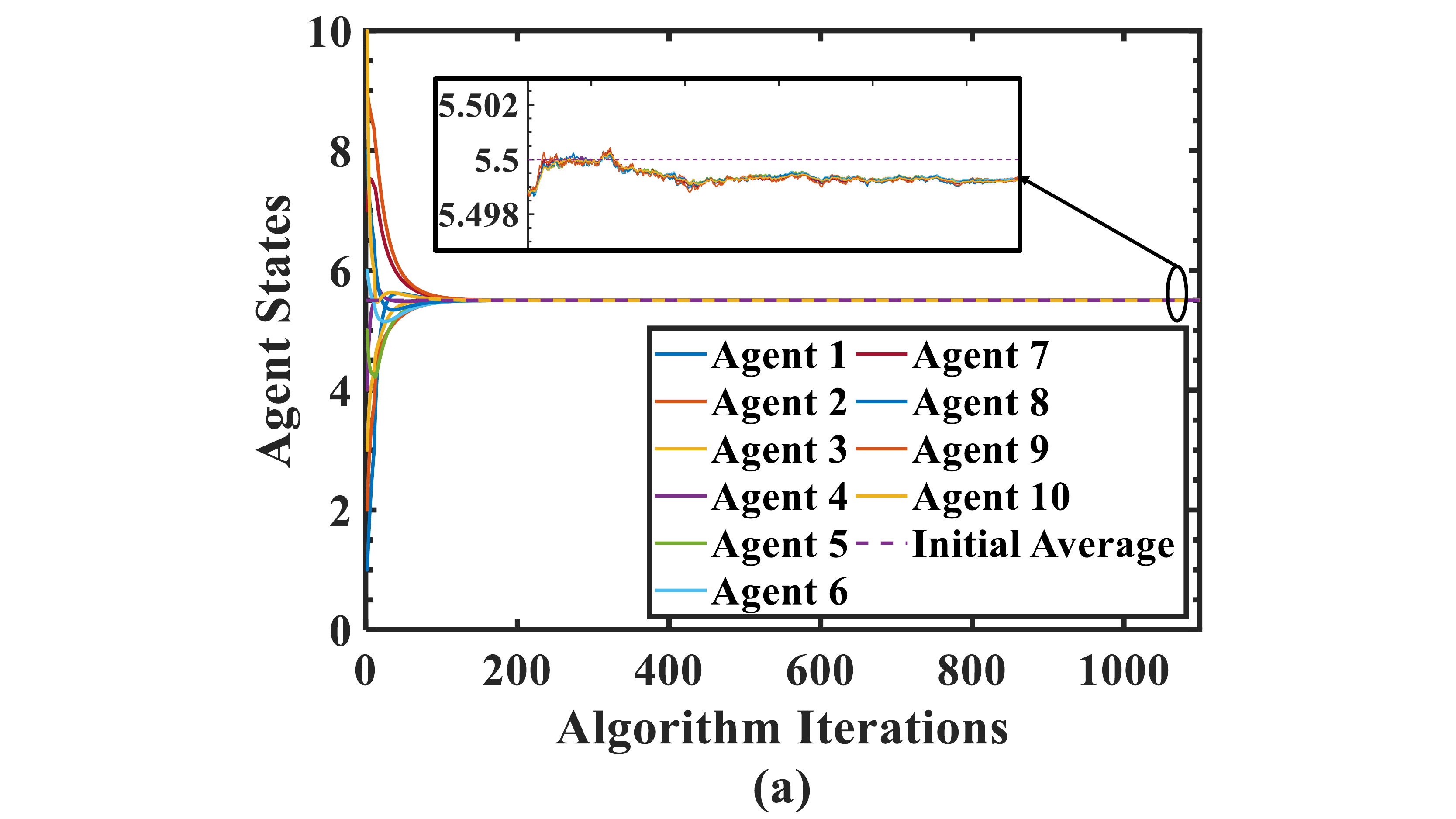}
    \includegraphics[scale=0.195,trim={6.1cm 0.2cm 5.8cm 0.45cm},clip] {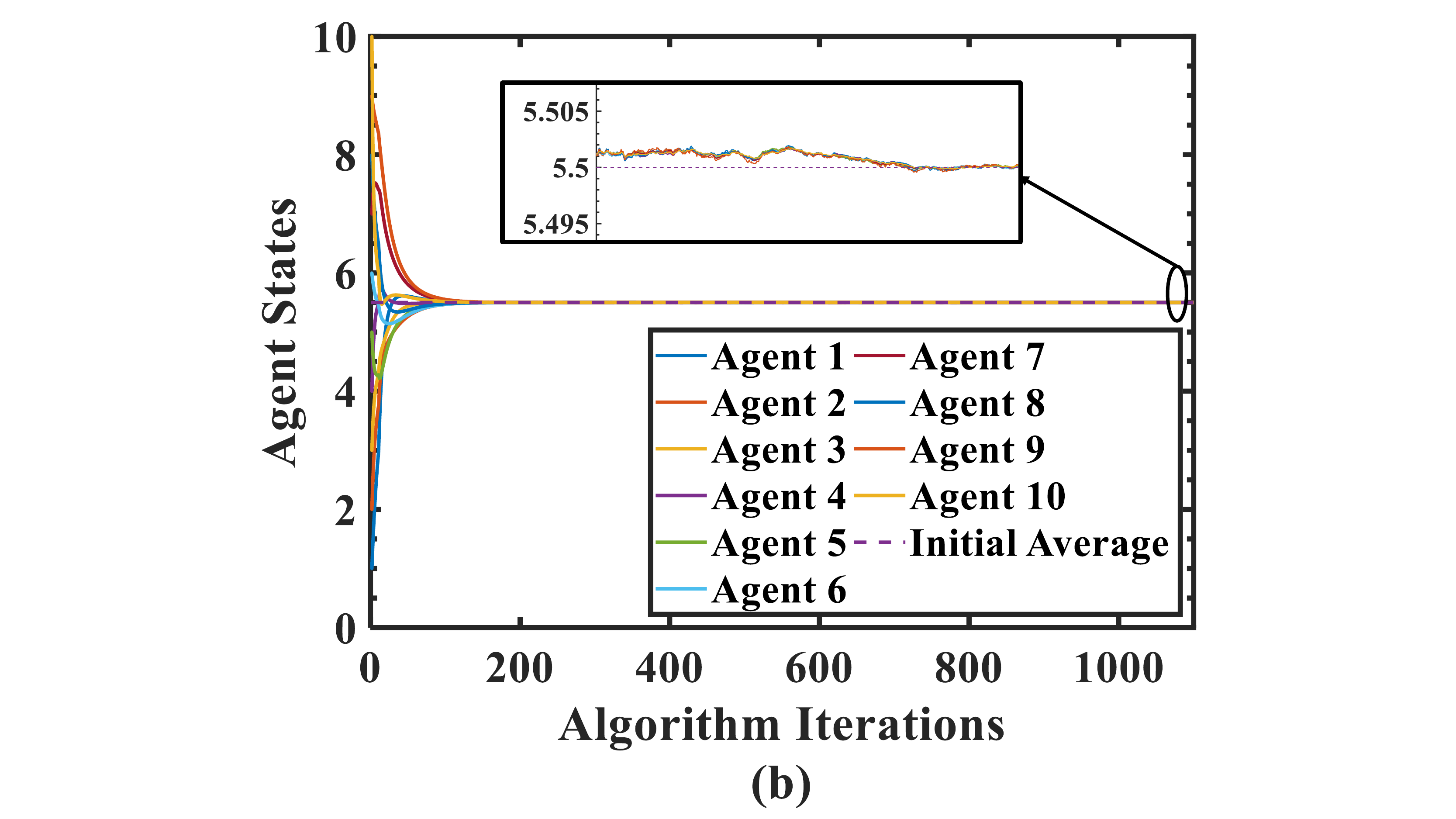}
  \caption{Performance of the $\nrps$ algorithm under the variation of the magnitude of additive communication noise: (a) noise realized via $\mathcal{U}(-5,5)$, (b) noise realized via $\mathcal{U}(-10,10)$. In both the cases, state estimates under $\nrps$ with appropriate algorithmic parameters converge to the initial average.}
  \label{fig:noise_mag} 
\end{figure}

\noindent \textbf{3. Robustness to noise with extreme magnitudes:} 
The parameters of the proposed $\nrps$ algorithm depend on an upper bound on the communication noise (see~(\ref{eq:theta_family})). To test the applicability of the $\nrps$ algorithm in extreme scenarios, here we test the robustness of the $\nrps$ algorithm under situations when the communication noise violates the upper bound in Assumption~\ref{assp:noise}. In particular, we design the algorithm parameters assuming that the communication noise is a uniform random variable $\mathcal{U}(-1,1)$. We test the following two scenarios: (i) communication noise in the worst case becomes $400$ times at every $50$ iterations, i.e. noise is realized as $\mathcal{U}(-400,400)$ at every iteration of the form $50i$, where, $i=1,2,\dots$, (ii) communication noise in the worst case becomes $400$ times at every $10$ iterations, i.e. noise is realized as $\mathcal{U}(-400,400)$ at every iteration of the form $10i$, where, $i=1,2,\dots$. The two situations are designed to test the effect of the number of violations of the bound in Assumption~\ref{assp:noise} during run-time on the convergence of the $\nrps$ algorithm. Figs.~\ref{fig:robustness}(a) and~\ref{fig:robustness}(b) show the agent states for the two test situations. The frequency of the event when the noise value is greater than the original upper bound effects the final converged value of the algorithm. In case (i) (Fig.~\ref{fig:robustness}(a)) with less number of violations the $\nrps$ algorithm converge closer to the initial average value of $5.5$, while the estimates in situation (ii) have a larger deviation from the original average value of $5.5$ (Fig.~\ref{fig:robustness}(b)). However, we emphasize that the proposed $\nrps$ algorithm performs well and remain stable under the extreme violations of the design assumptions as shown by the state trajectories in Figs.~\ref{fig:robustness}(a) and~\ref{fig:robustness}(b). 

\begin{figure}[t] 
\centering
    \includegraphics[scale=0.195,trim={5.8cm 0.2cm 5.9cm 0.2cm},clip] {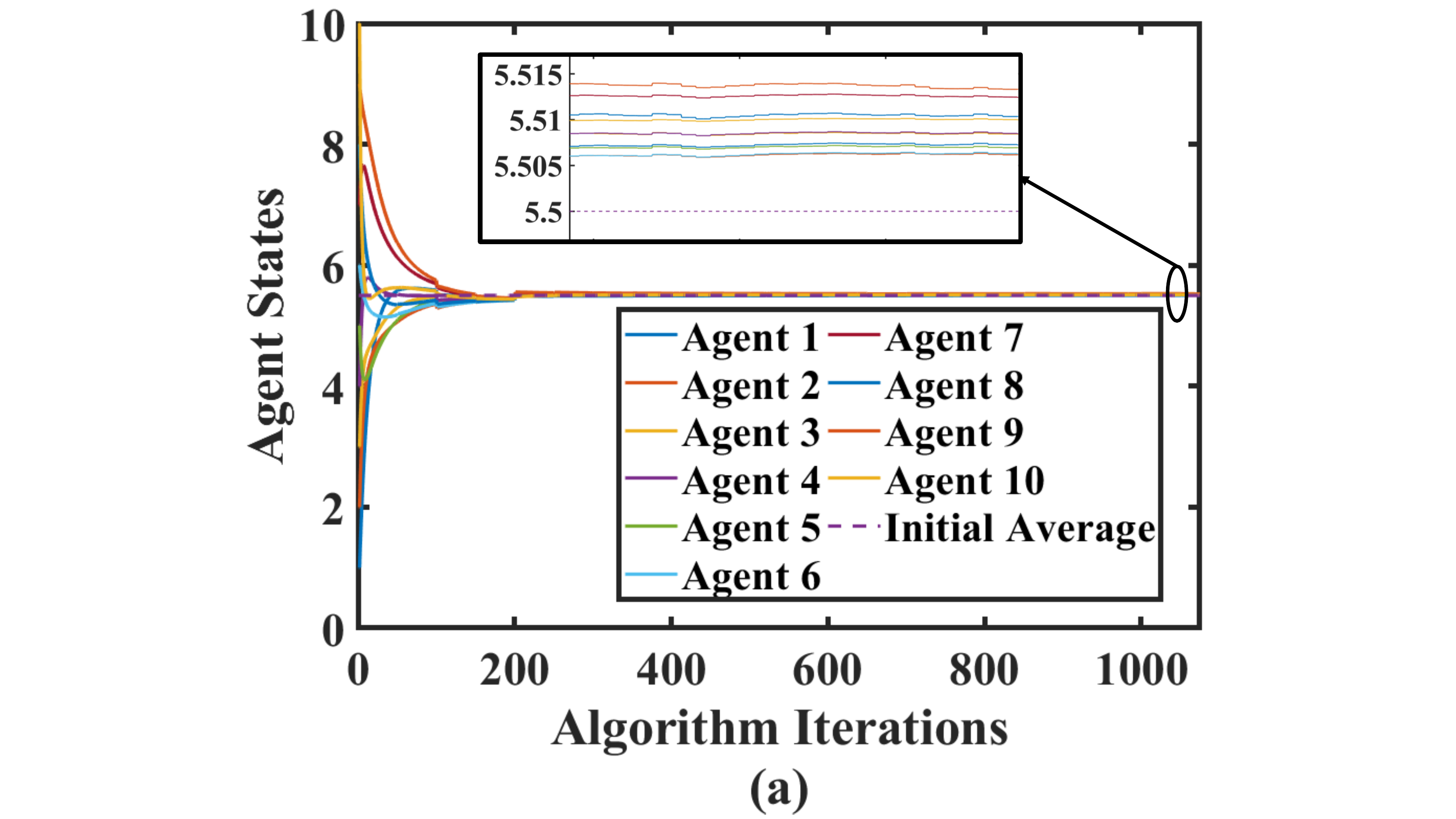}
    \includegraphics[scale=0.195,trim={5.9cm 0.2cm 5.6cm 0.18cm},clip] {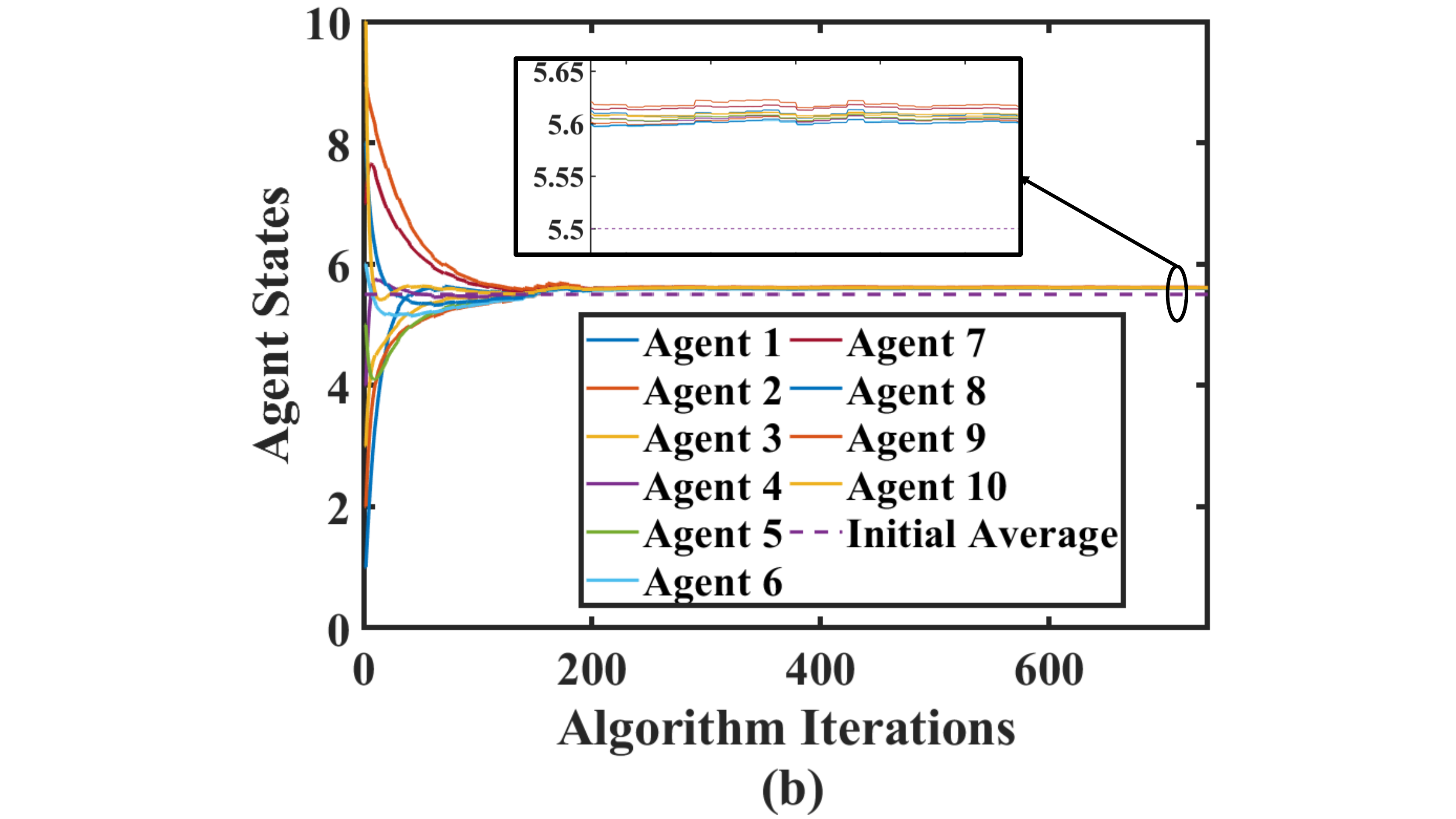}
  \caption{Performance of the $\nrps$ algorithm under scenarios with the noise realizations violating the upper bound utilized to the algorithms parameters. Figs.~\ref{fig:robustness}(a) and~\ref{fig:robustness}(b) demonstrate the robustness of the proposed $\nrps$ algorithm under extreme violation scenarios. }
  \label{fig:robustness} 
\end{figure}

\begin{figure}[b] 
\centering
    \includegraphics[scale=0.21,trim={0.6cm 5.2cm 0.8cm 6cm},clip] {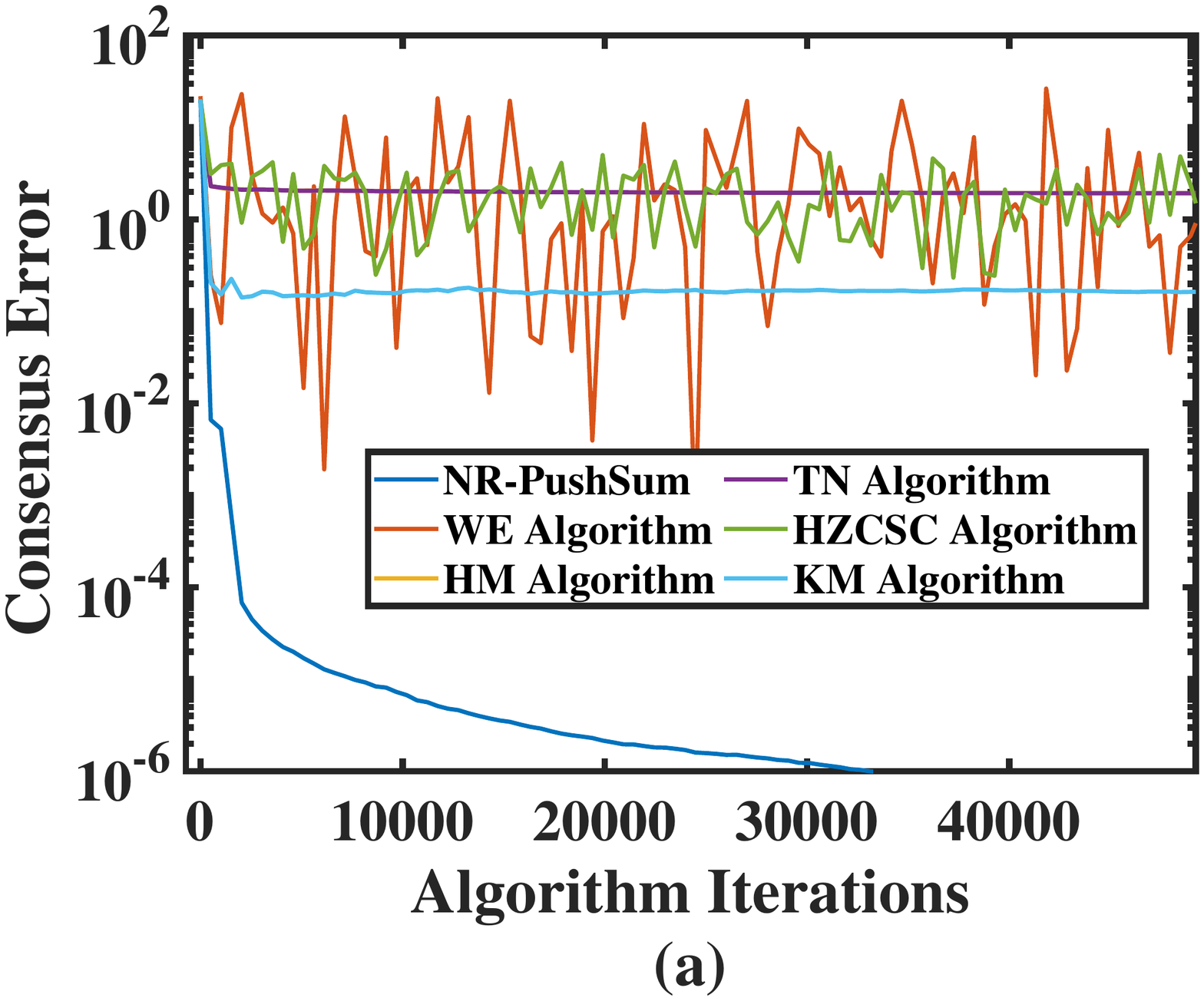}
    \includegraphics[scale=0.21,trim={0.6cm 5.2cm 0.8cm 6cm},clip] {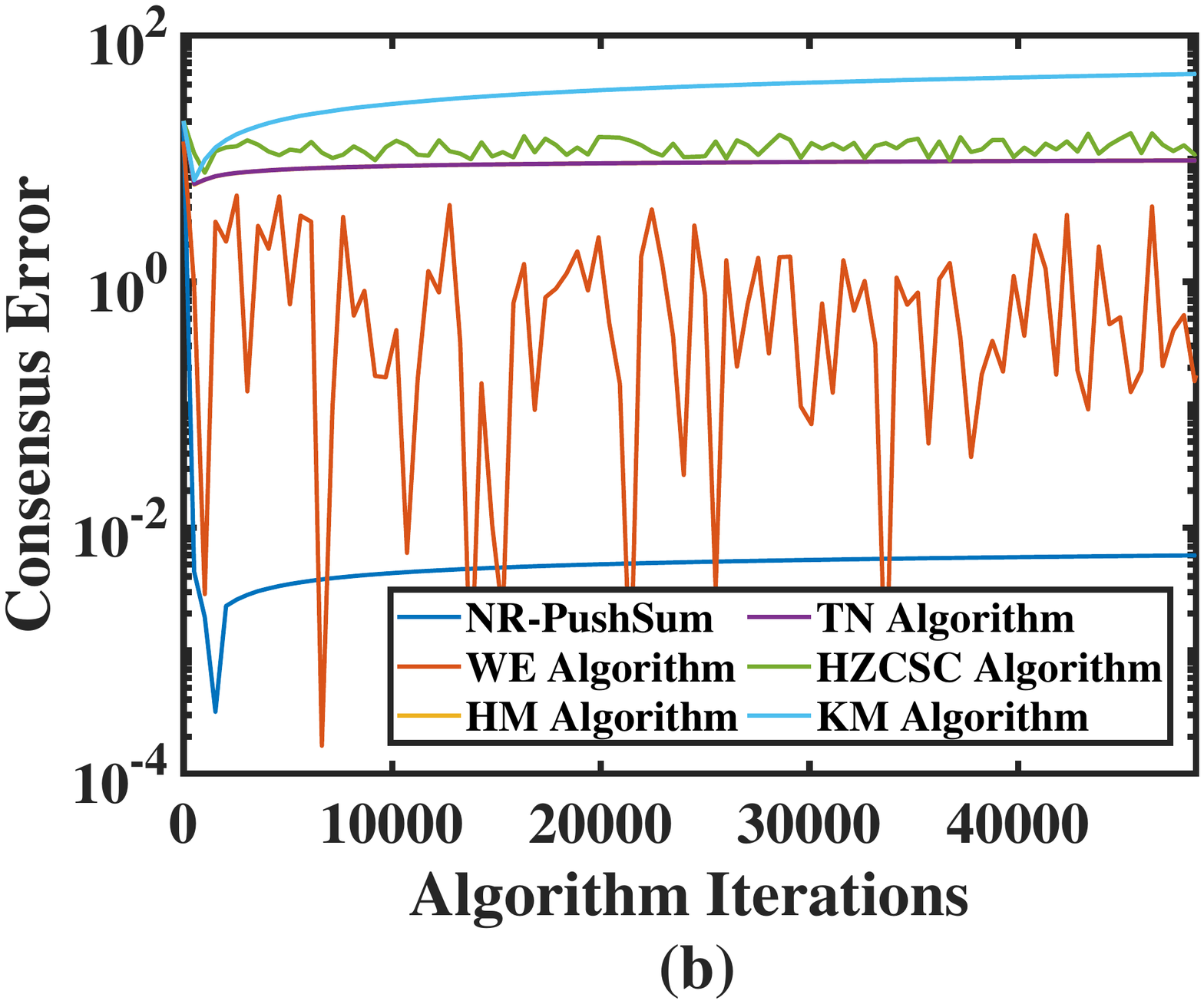}
  \caption{Performance comparison of the $\nrps$ algorithm with existing algorithms in the literature: (a) performance under bounded noise with zero mean and (b) performance under bounded communication noise with non-zero mean.}
  \label{fig:comparison} 
\end{figure}

\subsection{Comparison with other algorithms}
In this section, we present a performance comparison between the proposed $\nrps$ algorithm and existing algorithms in the literature. Here, we refer the existing algorithms via the initial of the authors' names. We compare with the following algorithms: WE algorithm of \cite{wang2013distributed}, HM algorithm in \cite{huang2009coordination}, TN algorithm in \cite{touri2009distributed}, HZCSC algorithm in \cite{he_bounded}, and the KM algorithm in \cite{kar2008distributed}. Consensus error for the $\nrps$ algorithm at any iteration $k$ is defined as, $\sum_{i=1}^n(z_i(k) - \frac{1}{n}\sum_{j=1}^n u_j)^2 = \sum_{i=1}(z_i(k) - 5.5)^2$. To obtain the consensus error for any other algorithm we replace the $z_i(k)$ estimates with the corresponding estimates in the chosen algorithm. We compare the results under two different noise realizations: (i) we first consider the communication noise realized as a uniform random variable $\mathcal{U}(-1,1)$. (ii) Next, we compare the performance of all the algorithms when the communication noise has a \textbf{non-zero} mean. In particular, we consider the noise realizations to be generated as a non-zero mean uniform random variable $\mathcal{U}(0,1)$. In both the cases we choose the following sequences for the $\nrps$ algorithm: $\beta(k) = 0.2, \forall k < 500, \beta(k) = 1/k^{1.5}, \forall k \geq 500 $ and $\theta(k) = 100, \forall k < 500,\theta(k) = 10/k^{1.5},  \forall k \geq 500$. To present the performance comparison, we plot the consensus error for all the algorithms with respect to the iteration count under the zero mean and non-zero mean cases in Fig.~\ref{fig:comparison}(a) and~\ref{fig:comparison}(b) respectively. It can be seen that the proposed $\nrps$ algorithm performs the best in both cases. The $\nrps$ algorithm has a consensus error that decays faster compared to the other algorithms when the communication noise has zero mean (see Fig.~\ref{fig:comparison}(a)). When the communication noise has non-zero mean, it can be seen that the consensus error for the existing algorithms in the literature is greater than $1$ and increases with the iteration count of the algorithm. In this case too the proposed $\nrps$ algorithm performs well. The consensus error for the $\nrps$ algorithm is less than $10^{-2}$ for all the iterations.

Summarizing, the simulation results indicate that the proposed $\nrps$ algorithm has a good performance under communication noise and show robustness to uncertain change in the communication noise. Compared to the existing state-of-the-art algorithms the $\nrps$ algorithm have superior performance with respect to the ability to achieve consensus to the initial average under communication noise.

\end{section}
\begin{section}{Conclusion}\label{sec:conclusion}
In this article, we consider the problem of designing an average consensus algorithm over directed graphs when the communication links are corrupted by noise. First we show that the widely used PushSum algorithm to design distributed decision making schemes over directed graphs does not work well under communication noise. Then we propose the $\nrps$ algorithm where each agent utilizes its own noise-free weighted initial value to update its estimates during every iteration. Further, each agent controls the amount of noise to be allowed into its state updates by appropriately weighting the information sent by its neighbors over the noisy communication links. The carefully chosen weights allow the agents to compensate for the noisy estimates received from their neighbors and steer towards consensus. We establish that the proposed $\nrps$ algorithm achieves consensus at a geometric rate in the noiseless communication case. We prove that under additive noise the agents' estimates following the $\nrps$ algorithm reach to a consensus value \textit{almost surely}. The efficacy of the proposed algorithm is demonstrated via numerical case studies involving different noise realizations and various combinations of the algorithm parameters. Utilizing the proposed $\nrps$ algorithm to design distributed optimization algorithms over directed graphs in the presence of communication noise is a future research endeavor by the authors.
\end{section}

\appendix
\subsection{Proof of matrix $W(l,k)$ being cSIA for all $l,k$}\label{sec:W_k_cSIA}
Recall $\M(k) := \A(k) + \B(k)$. Consider the column sum, for any column $j$,
\begin{align*}
    \sum_{i=1}^n \M_{ij}(k) &= \alpha_{ii}(k) p_{ii} + \beta(k) \sum_{i=1, i \neq j}^n p_{ij}\\
    &= \frac{1 - \beta(k)(1-p_{ii})}{p_{ii}} p_{ii} + \beta(k) \sum_{i=1, i \neq j}^n p_{ij} \\
    & = 1 - \beta(k)(1-p_{ii}) + \beta(k) (1 - p_{ii}) = 1,
\end{align*}
where, we used~(\ref{eq:alpha}) in the second equality and the fact that matrix $\p$ satisfy Assumption~\ref{assp:strg_colmstoc}. Since $j$ is arbitrary, we conclude that $\M(k)$ is column-stochastic for all $k$.

\textit{Claim:} Let $R$ and $S$ be two $n \times n$ column stochastic matrices. Then $RS \in \mathbb{R}^{n\times n}$ is column stochastic. \\
\textit{Proof:}
Consider, the sum
\begin{align*}
    \sum_{i = 1}^n [RS]_{ij} &= \sum_{i=1}^n \sum_{k=1}^n R_{ik}S_{kj} \\
    &=  \sum_{k=1}^n \sum_{i=1}^n R_{ik} S_{kj} \\
    &= \sum_{k=1}^n S_{kj} = 1,
\end{align*}
where, we used the fact that $R$ and $S$ are column-stochastic. Thus, $RS$ is column-stochastic. \qed

Note that,
\begin{align*}
    W_{ij}(l,k) &= [W(l,k-1)\M(k)]_{ij} \\
    &= [W(l,k-1)(\A(k) + \B(k))]_{ij} \\
    &= [W(l,k-1)\A(k)]_{ij} + [W(l,k-1)\B(k)]_{ij} \\
    & \hspace{-0.25in}= \sum_{t = 1}^n W_{it}(l,k-1)\A_{tj}(k) + \sum_{t = 1}^n W_{it}(l,k-1)\B_{tj}(k)\\
    & \hspace{-0.25in} = W_{ij}(l,k-1)\alpha_{j}(k)p_{jj} + \beta(k) \sum_{t = 1, t\neq j}^n W_{it}(l,k-1)p_{tj}\\
    & \hspace{-0.2in} = W_{ij}(l,k-1)\alpha_{j}(k)p_{jj} + 
    W_{ii}(l,k-1)\beta(k)p_{ij} \\ 
    & \hspace{0.4in} + \beta(k) \sum_{t = 1, t \neq i, t\neq j}^n W_{it}(l,k-1)p_{tj}.
\end{align*}
This implies, if $p_{ij} \neq 0$ then $W_{ij}(l,k) \neq 0$. Thus, the graph structure generated using $W(l,k)$ is strongly connected as $\mathcal{G}(\mathcal{V},\mathcal{E})$ (graph generated using $\p$) is strongly connected. Hence, $W(l,k)$ is irreducible. As, $W_{jj}(l,k) > 0$ for all $j$ and the graph induced by $W(l,k)$ is strongly connected, $W(l.k)$ is aperiodic. 
\bibliography{references}

\begin{thebibliography}{10}
\providecommand{\url}[1]{#1}
\csname url@samestyle\endcsname
\providecommand{\newblock}{\relax}
\providecommand{\bibinfo}[2]{#2}
\providecommand{\BIBentrySTDinterwordspacing}{\spaceskip=0pt\relax}
\providecommand{\BIBentryALTinterwordstretchfactor}{4}
\providecommand{\BIBentryALTinterwordspacing}{\spaceskip=\fontdimen2\font plus
\BIBentryALTinterwordstretchfactor\fontdimen3\font minus
  \fontdimen4\font\relax}
\providecommand{\BIBforeignlanguage}[2]{{%
\expandafter\ifx\csname l@#1\endcsname\relax
\typeout{** WARNING: IEEEtran.bst: No hyphenation pattern has been}%
\typeout{** loaded for the language `#1'. Using the pattern for}%
\typeout{** the default language instead.}%
\else
\language=\csname l@#1\endcsname
\fi
#2}}
\providecommand{\BIBdecl}{\relax}
\BIBdecl

\bibitem{patel2017distributed}
S.~Patel, S.~Attree, S.~Talukdar, M.~Prakash, and M.~V. Salapaka, ``Distributed
  apportioning in a power network for providing demand response services,'' in
  \emph{2017 IEEE International Conference on Smart Grid Communications
  (SmartGridComm)}.\hskip 1em plus 0.5em minus 0.4em\relax IEEE, 2017, pp.
  38--44.

\bibitem{patel2020distributed}
S.~Patel, B.~Lundstrom, G.~Saraswat, and M.~V. Salapaka, ``Distributed power
  apportioning with early dispatch for ancillary services in renewable grids,''
  \emph{arXiv preprint arXiv:2007.11715}, 2020.

\bibitem{patel2020codit}
S.~Patel, V.~Khatana, G.~Saraswat, and M.~V. Salapaka, ``Distributed detection
  of malicious attacks on consensus algorithms with applications in power
  networks,'' in \emph{2020 7th International Conference on Control, Decision
  and Information Technologies (CoDIT)}, vol.~1.\hskip 1em plus 0.5em minus
  0.4em\relax IEEE, 2020, pp. 397--402.

\bibitem{ala2016classifiers}
M.~Ala'raj and M.~F. Abbod, ``Classifiers consensus system approach for credit
  scoring,'' \emph{Knowledge-Based Systems}, vol. 104, pp. 89--105, 2016.

\bibitem{FaxMur04}
A.~Fax and R.~M. Murray, ``Information flow and cooperative control of vehicle
  formations,'' \emph{IEEE Transactions on Automatic Control}, vol.~49, no.~9,
  pp. 1465--1476, 2004.

\bibitem{olfati2007consensus}
R.~Olfati-Saber, J.~A. Fax, and R.~M. Murray, ``Consensus and cooperation in
  networked multi-agent systems,'' \emph{Proceedings of the IEEE}, vol.~95,
  no.~1, pp. 215--233, 2007.

\bibitem{nedic2020distributed}
A.~Nedic, ``Distributed gradient methods for convex machine learning problems
  in networks: Distributed optimization,'' \emph{IEEE Signal Processing
  Magazine}, vol.~37, no.~3, pp. 92--101, 2020.

\bibitem{arrow1958decentralization}
K.~J. Arrow and L.~Hurwicz, \emph{Decentralization and computation in resource
  allocation}.\hskip 1em plus 0.5em minus 0.4em\relax Stanford University,
  Department of Economics, 1958.

\bibitem{degroot1974reaching}
M.~H. DeGroot, ``Reaching a consensus,'' \emph{Journal of the American
  Statistical Association}, vol.~69, no. 345, pp. 118--121, 1974.

\bibitem{lynch1996distributed}
N.~A. Lynch, \emph{Distributed algorithms}.\hskip 1em plus 0.5em minus
  0.4em\relax Elsevier, 1996.

\bibitem{jadbabaie2003coordination}
A.~Jadbabaie, J.~Lin, and A.~S. Morse, ``Coordination of groups of mobile
  autonomous agents using nearest neighbor rules,'' \emph{Departmental Papers
  (ESE)}, p.~29, 2003.

\bibitem{olfati2004consensus}
R.~Olfati-Saber and R.~M. Murray, ``Consensus problems in networks of agents
  with switching topology and time-delays,'' \emph{IEEE Transactions on
  automatic control}, vol.~49, no.~9, pp. 1520--1533, 2004.

\bibitem{kempe2003gossip}
D.~Kempe, A.~Dobra, and J.~Gehrke, ``Gossip-based computation of aggregate
  information,'' in \emph{44th Annual IEEE Symposium on Foundations of Computer
  Science, 2003. Proceedings.}\hskip 1em plus 0.5em minus 0.4em\relax IEEE,
  2003, pp. 482--491.

\bibitem{hadjicostis2012average}
C.~N. Hadjicostis and T.~Charalambous, ``Average consensus in the presence of
  delays and dynamically changing directed graph topologies,'' \emph{arXiv
  preprint arXiv:1210.4778}, 2012.

\bibitem{saraswat2019distributed}
G.~Saraswat, V.~Khatana, S.~Patel, and M.~V. Salapaka, ``Distributed
  finite-time termination for consensus algorithm in switching topologies,''
  \emph{arXiv preprint arXiv:1909.00059}, 2019.

\bibitem{hadjicostis2013average}
C.~N. Hadjicostis and T.~Charalambous, ``Average consensus in the presence of
  delays in directed graph topologies,'' \emph{IEEE Transactions on Automatic
  Control}, vol.~59, no.~3, pp. 763--768, 2013.

\bibitem{mangalJournal}
M.~Prakash, S.~Talukdar, S.~Attree, V.~Yadav, and M.~V. Salapaka, ``Distributed
  stopping criterion for consensus in the presence of delays,'' \emph{IEEE
  Transactions on Control of Network Systems}, vol.~7, no.~1, pp. 85--95, 2020.

\bibitem{melbourne2020geometry}
J.~Melbourne, G.~Saraswat, V.~Khatana, S.~Patel, and M.~V. Salapaka, ``On the
  geometry of consensus algorithms with application to distributed termination
  in higher dimension,'' \emph{IFAC-PapersOnLine}, vol.~53, no.~2, pp.
  2951--2956, 2020.

\bibitem{melbourne2020convex}
------, ``Convex decreasing algorithms: Distributed synthesis and finite-time
  termination in higher dimension,'' \emph{arXiv preprint arXiv:2007.13050},
  2020.

\bibitem{liu2011distributed}
S.~Liu, L.~Xie, and H.~Zhang, ``Distributed consensus for multi-agent systems
  with delays and noises in transmission channels,'' \emph{Automatica},
  vol.~47, no.~5, pp. 920--934, 2011.

\bibitem{xiao2007distributed}
L.~Xiao, S.~Boyd, and S.-J. Kim, ``Distributed average consensus with
  least-mean-square deviation,'' \emph{Journal of parallel and distributed
  computing}, vol.~67, no.~1, pp. 33--46, 2007.

\bibitem{carli2009average}
R.~Carli, G.~Como, P.~Frasca, and F.~Garin, ``Average consensus on digital
  noisy networks,'' \emph{IFAC Proceedings Volumes}, vol.~42, no.~20, pp.
  36--41, 2009.

\bibitem{pan2015consensus}
T.~Pan, L.~Mo, and X.~Cao, ``Consensus of discrete-time multi-agent systems
  with white noise disturbance,'' \emph{IFAC-PapersOnLine}, vol.~48, no.~28,
  pp. 202--205, 2015.

\bibitem{hanada2020stochastic}
K.~Hanada, T.~Wada, I.~Masubuchi, T.~Asai, and Y.~Fujisaki, ``Stochastic
  consensus algorithms over general noisy networks,'' \emph{SICE Journal of
  Control, Measurement, and System Integration}, vol.~13, no.~6, pp. 274--281,
  2020.

\bibitem{dasarathan2015robust}
S.~Dasarathan, C.~Tepedelenlio{\u{g}}lu, M.~K. Banavar, and A.~Spanias,
  ``Robust consensus in the presence of impulsive channel noise,'' \emph{IEEE
  Transactions on Signal Processing}, vol.~63, no.~8, pp. 2118--2129, 2015.

\bibitem{li2018analysis}
B.~Li, H.~Leung, and C.~Seneviratne, ``Analysis of noise impact on distributed
  average consensus,'' in \emph{Signal Processing, Sensor/Information Fusion,
  and Target Recognition XXVII}, vol. 10646.\hskip 1em plus 0.5em minus
  0.4em\relax International Society for Optics and Photonics, 2018, p. 106460L.

\bibitem{huang2007stochastic}
M.~Huang and J.~H. Manton, ``Stochastic approximation for consensus seeking:
  Mean square and almost sure convergence,'' in \emph{2007 46th IEEE Conference
  on Decision and Control}.\hskip 1em plus 0.5em minus 0.4em\relax IEEE, 2007,
  pp. 306--311.

\bibitem{kar2007distributed}
S.~Kar and J.~M. Moura, ``Distributed average consensus in sensor networks with
  random link failures and communication channel noise,'' in \emph{2007
  Conference Record of the Forty-First Asilomar Conference on Signals, Systems
  and Computers}.\hskip 1em plus 0.5em minus 0.4em\relax IEEE, 2007, pp.
  676--680.

\bibitem{kar2008distributed}
------, ``Distributed consensus algorithms in sensor networks with imperfect
  communication: Link failures and channel noise,'' \emph{IEEE Transactions on
  Signal Processing}, vol.~57, no.~1, pp. 355--369, 2008.

\bibitem{pescosolido2008average}
L.~Pescosolido, S.~Barbarossa, and G.~Scutari, ``Average consensus algorithms
  robust against channel noise,'' in \emph{2008 IEEE 9th Workshop on Signal
  Processing Advances in Wireless Communications}.\hskip 1em plus 0.5em minus
  0.4em\relax IEEE, 2008, pp. 261--265.

\bibitem{li2010consensus}
T.~Li and J.-F. Zhang, ``Consensus conditions of multi-agent systems with
  time-varying topologies and stochastic communication noises,'' \emph{IEEE
  Transactions on Automatic Control}, vol.~55, no.~9, pp. 2043--2057, 2010.

\bibitem{aysal2010convergence}
T.~C. Aysal and K.~E. Barner, ``Convergence of consensus models with stochastic
  disturbances,'' \emph{IEEE Transactions on Information Theory}, vol.~56,
  no.~8, pp. 4101--4113, 2010.

\bibitem{rajagopal2010network}
R.~Rajagopal and M.~J. Wainwright, ``Network-based consensus averaging with
  general noisy channels,'' \emph{IEEE Transactions on Signal Processing},
  vol.~59, no.~1, pp. 373--385, 2010.

\bibitem{schizas2007consensus}
I.~D. Schizas, A.~Ribeiro, and G.~B. Giannakis, ``Consensus in ad hoc wsns with
  noisy links—part i: Distributed estimation of deterministic signals,''
  \emph{IEEE Transactions on Signal Processing}, vol.~56, no.~1, pp. 350--364,
  2007.

\bibitem{schizas2008consensus}
I.~D. Schizas, G.~B. Giannakis, S.~I. Roumeliotis, and A.~Ribeiro, ``Consensus
  in ad hoc wsns with noisy links—part ii: Distributed estimation and
  smoothing of random signals,'' \emph{IEEE Transactions on Signal Processing},
  vol.~56, no.~4, pp. 1650--1666, 2008.

\bibitem{huang2009coordination}
M.~Huang and J.~H. Manton, ``Coordination and consensus of networked agents
  with noisy measurements: Stochastic algorithms and asymptotic behavior,''
  \emph{SIAM Journal on Control and Optimization}, vol.~48, no.~1, pp.
  134--161, 2009.

\bibitem{kibangou2011finite}
A.~Y. Kibangou, ``Finite-time average consensus based protocol for distributed
  estimation over awgn channels,'' in \emph{2011 50th IEEE Conference on
  Decision and Control and European Control Conference}.\hskip 1em plus 0.5em
  minus 0.4em\relax IEEE, 2011, pp. 5595--5600.

\bibitem{zhou2013discrete}
M.~Zhou, J.~He, P.~Cheng, and J.~Chen, ``Discrete average consensus with
  bounded noise,'' in \emph{52nd IEEE Conference on Decision and
  Control}.\hskip 1em plus 0.5em minus 0.4em\relax IEEE, 2013, pp. 5270--5275.

\bibitem{he_bounded}
J.~He, M.~Zhou, P.~Cheng, L.~Shi, and J.~Chen, ``Consensus under bounded noise
  in discrete network systems: An algorithm with fast convergence and high
  accuracy,'' \emph{IEEE Transactions on Cybernetics}, vol.~46, no.~12, pp.
  2874--2884, 2016.

\bibitem{rego2015consensus}
F.~F. Rego, Y.~Pu, A.~Alessandretti, A.~P. Aguiar, and C.~N. Jones, ``A
  consensus algorithm for networks with process noise and quantization error,''
  in \emph{2015 53rd Annual Allerton Conference on Communication, Control, and
  Computing (Allerton)}.\hskip 1em plus 0.5em minus 0.4em\relax IEEE, 2015, pp.
  488--495.

\bibitem{jadbabaie2016performance}
A.~Jadbabaie and A.~Olshevsky, ``On performance of consensus protocols subject
  to noise: Role of hitting times and network structure,'' in \emph{2016 IEEE
  55th Conference on Decision and Control (CDC)}.\hskip 1em plus 0.5em minus
  0.4em\relax IEEE, 2016, pp. 179--184.

\bibitem{morral2017success}
G.~Morral, P.~Bianchi, and G.~Fort, ``Success and failure of
  adaptation-diffusion algorithms with decaying step size in multiagent
  networks,'' \emph{IEEE Transactions on Signal Processing}, vol.~65, no.~11,
  pp. 2798--2813, 2017.

\bibitem{chen2017critical}
G.~Chen, C.~Chen, G.~Yin \emph{et~al.}, ``Critical connectivity and fastest
  convergence rates of distributed consensus with switching topologies and
  additive noises,'' \emph{IEEE Transactions on Automatic Control}, vol.~62,
  no.~12, pp. 6152--6167, 2017.

\bibitem{granichin2020simultaneous}
O.~Granichin, V.~Erofeeva, Y.~Ivanskiy, and Y.~Jiang, ``Simultaneous
  perturbation stochastic approximation-based consensus for tracking under
  unknown-but-bounded disturbances,'' \emph{IEEE Transactions on Automatic
  Control}, vol.~66, no.~8, pp. 3710--3717, 2020.

\bibitem{mateos2016noise}
D.~Mateos-N{\'u}nez and J.~Cort{\'e}s, ``Noise-to-state exponentially stable
  distributed convex optimization on weight-balanced digraphs,'' \emph{SIAM
  Journal on Control and Optimization}, vol.~54, no.~1, pp. 266--290, 2016.

\bibitem{touri2009distributed}
B.~Touri and A.~Nedic, ``Distributed consensus over network with noisy links,''
  in \emph{2009 12th International Conference on Information Fusion}.\hskip 1em
  plus 0.5em minus 0.4em\relax IEEE, 2009, pp. 146--154.

\bibitem{wang2010dynamic}
J.~Wang and N.~Elia, ``Dynamic average consensus over random networks with
  additive noise,'' in \emph{49th IEEE Conference on Decision and Control
  (CDC)}.\hskip 1em plus 0.5em minus 0.4em\relax IEEE, 2010, pp. 4789--4794.

\bibitem{wang2013distributed}
------, ``Distributed averaging algorithms resilient to communication noise and
  dropouts,'' \emph{IEEE transactions on signal processing}, vol.~61, no.~9,
  pp. 2231--2242, 2013.

\bibitem{pu2018flocking}
S.~Pu and A.~Garcia, ``A flocking-based approach for distributed stochastic
  optimization,'' \emph{Operations Research}, vol.~66, no.~1, pp. 267--281,
  2018.

\bibitem{long2015distributed}
Y.~Long, S.~Liu, and L.~Xie, ``Distributed consensus of discrete-time
  multi-agent systems with multiplicative noises,'' \emph{International Journal
  of Robust and Nonlinear Control}, vol.~25, no.~16, pp. 3113--3131, 2015.

\bibitem{sheipak2020reaching}
S.~Sheipak, ``Reaching consensus in a variable-topology multiagent system under
  additive random noise,'' \emph{Automation and Remote Control}, vol.~81,
  no.~5, pp. 911--921, 2020.

\bibitem{huang2009stochastic}
M.~Huang and J.~H. Manton, ``Stochastic consensus seeking with noisy and
  directed inter-agent communication: fixed and randomly varying topologies,''
  \emph{IEEE Transactions on Automatic Control}, vol.~55, no.~1, pp. 235--241,
  2009.

\bibitem{wang2015consensus}
Y.~Wang, L.~Cheng, Z.-G. Hou, M.~Tan, C.~Zhou, and M.~Wang, ``Consensus seeking
  in a network of discrete-time linear agents with communication noises,''
  \emph{International Journal of Systems Science}, vol.~46, no.~10, pp.
  1874--1888, 2015.

\bibitem{li2009mean}
T.~Li and J.-F. Zhang, ``Mean square average-consensus under measurement noises
  and fixed topologies: Necessary and sufficient conditions,''
  \emph{Automatica}, vol.~45, no.~8, pp. 1929--1936, 2009.

\bibitem{yaziciouglu2020high}
Y.~Yaz{\i}c{\i}o{\u{g}}lu and A.~Speranzon, ``High dimensional robust consensus
  over networks with limited capacity,'' \emph{IEEE Control Systems Letters},
  vol.~5, no.~6, pp. 2024--2029, 2020.

\bibitem{wang2009distributed}
J.~Wang and N.~Elia, ``Distributed agreement in the presence of noise,'' in
  \emph{2009 47th Annual Allerton Conference on Communication, Control, and
  Computing (Allerton)}.\hskip 1em plus 0.5em minus 0.4em\relax IEEE, 2009, pp.
  1575--1581.

\bibitem{wang2016robust}
Z.~Wang, W.~Wang, and H.~Zhang, ``Robust consensus for linear multi-agent
  systems with noises,'' \emph{IET Control Theory \& Applications}, vol.~10,
  no.~17, pp. 2348--2356, 2016.

\bibitem{zong2015stochastic}
X.~Zong, T.~Li, and J.-F. Zhang, ``Stochastic consensus of continuous-time
  multi-agent systems with additive measurement noises,'' in \emph{2015 54th
  IEEE Conference on Decision and Control (CDC)}.\hskip 1em plus 0.5em minus
  0.4em\relax IEEE, 2015, pp. 543--548.

\bibitem{cheng2013mean}
L.~Cheng, Z.-G. Hou, and M.~Tan, ``A mean square consensus protocol for linear
  multi-agent systems with communication noises and fixed topologies,''
  \emph{IEEE Transactions on Automatic Control}, vol.~59, no.~1, pp. 261--267,
  2013.

\bibitem{khatana2020gradient}
V.~Khatana, G.~Saraswat, S.~Patel, and M.~V. Salapaka, ``Gradient-consensus
  method for distributed optimization in directed multi-agent networks,'' in
  \emph{2020 American Control Conference (ACC)}.\hskip 1em plus 0.5em minus
  0.4em\relax IEEE, 2020, pp. 4689--4694.

\bibitem{khatana2020d}
V.~Khatana and M.~V. Salapaka, ``D-distadmm: Ao (1/k) distributed admm for
  distributed optimization in directed graph topologies,'' in \emph{2020 59th
  IEEE Conference on Decision and Control (CDC)}.\hskip 1em plus 0.5em minus
  0.4em\relax IEEE, 2020, pp. 2992--2997.

\bibitem{nedic2014distributed}
A.~Nedi{\'c} and A.~Olshevsky, ``Distributed optimization over time-varying
  directed graphs,'' \emph{IEEE Transactions on Automatic Control}, vol.~60,
  no.~3, pp. 601--615, 2014.

\bibitem{nedic2017achieving}
A.~Nedic, A.~Olshevsky, and W.~Shi, ``Achieving geometric convergence for
  distributed optimization over time-varying graphs,'' \emph{SIAM Journal on
  Optimization}, vol.~27, no.~4, pp. 2597--2633, 2017.

\bibitem{pu2020push}
S.~Pu, W.~Shi, J.~Xu, and A.~Nedi{\'c}, ``Push--pull gradient methods for
  distributed optimization in networks,'' \emph{IEEE Transactions on Automatic
  Control}, vol.~66, no.~1, pp. 1--16, 2020.

\bibitem{khatana2019gradient}
V.~Khatana, G.~Saraswat, S.~Patel, and M.~V. Salapaka, ``Gradient-consensus:
  Linearly convergent distributed optimization algorithm over directed
  graphs,'' \emph{arXiv preprint arXiv:1909.10070}, 2019.

\bibitem{khatana2020dc}
V.~Khatana and M.~V. Salapaka, ``Dc-distadmm: Admm algorithm for contsrained
  distributed optimization over directed graphs,'' \emph{arXiv preprint
  arXiv:2003.13742}, 2020.

\bibitem{Die06}
R.~Diestel, \emph{Graph Theory}.\hskip 1em plus 0.5em minus 0.4em\relax Berlin,
  Germany: Springer-Verlag, 2006.

\bibitem{horn2012matrix}
R.~A. Horn and C.~R. Johnson, \emph{Matrix analysis}.\hskip 1em plus 0.5em
  minus 0.4em\relax Cambridge university press, 2012.

\bibitem{brualdi1991combinatorial}
R.~A. Brualdi, H.~J. Ryser \emph{et~al.}, \emph{Combinatorial matrix
  theory}.\hskip 1em plus 0.5em minus 0.4em\relax Cambridge University Press,
  1991, no.~39.

\bibitem{Wolfowitz}
\BIBentryALTinterwordspacing
J.~Wolfowitz, ``Products of indecomposable, aperiodic, stochastic matrices,''
  \emph{Proceedings of the American Mathematical Society}, vol.~14, no.~5, pp.
  733--737, 1963. [Online]. Available:
  \url{http://www.jstor.org/stable/2034984}
\BIBentrySTDinterwordspacing

\bibitem{erdHos1960evolution}
P.~Erd{\H{o}}s and A.~R{\'e}nyi, ``On the evolution of random graphs,''
  \emph{Publ. Math. Inst. Hung. Acad. Sci}, vol.~5, no.~1, pp. 17--60, 1960.

\end{thebibliography}

\end{document}